\documentclass[a4paper,10pt,oneside]{article}
\usepackage[top=2.50cm, bottom=0.5cm,left=2.5cm, right=1.5cm]{geometry}
\usepackage{amssymb,amsmath,amsthm}
\usepackage[numbers]{natbib}
\usepackage{amsmath}
\date{}
\usepackage{graphicx}
\usepackage{caption}
\usepackage{subcaption}
\usepackage{tabularx}
\usepackage{multirow}
\usepackage{float}
\usepackage[colorlinks]{hyperref}
\usepackage{multicol}
\usepackage{adjustbox}
\usepackage{colortbl}
\usepackage{pdflscape}
\usepackage{algorithm}
\usepackage{algorithmic}
\usepackage{epstopdf}
\usepackage[affil-it]{authblk}
\title{\bf{Cubic lower record-based transmuted family of  distributions: Theory, Estimation, Applications}}
\author{Caner Tan{\i}\c{s}}\affil{Department of Statistics, Necmettin Erbakan University, Konya, Turkey;  caner.tanis@erbakan.edu.tr}

\newtheorem{theorem}{Theorem}

\begin{document}
\maketitle
\begin{abstract}

In this study, a family of distributions called cubic lower record-based transmuted is provided. A special case of this family is proposed as an alternative exponential distribution. Several statistical properties are explored. We utilize nine different methods to estimate the parameters of the suggested distribution. In order to compare the performances of these methods, we consider a comprehensive Monte-Carlo simulation study. As a result of simulation study, we conclude that minimum absolute distance estimator is a valuable alternative to maximum likelihood estimator. Then, we carried out two real-world data examples to evaluate the fits of introduced distribution as well as its potential competitor ones. The findings of real-world data analysis show that the best-fitting distribution for both datasets is our model.
\end{abstract}
\textbf{Keywords:}
Lower record values, Transmuted distributions, Point estimation, Monte Carlo simulation, Real data analysis.

\textbf{AMS:}
62E10, 62F10, 62P99.

\section{Introduction}
In last decades, the systematic proposition of lifetime distributions has become a key focus in the advancement of statistical modeling, especially by becoming flexibility and explicative power. Although classical distributions form the building blocks of probability theory, they often fall short when modeling real data sets exhibiting properties such as asymmetry, heavy tails, or multiple peaks. To overcome these limitations, interest has grown in transformation approaches and generalized distribution families. In this regard, the transmutation map introduced by Shaw and Buckley \cite{shaw2009alchemy} using a method based on order statistics is considered one of the most effective tools in the literature. The mathematical expression of the transmutation map defined in the relevant study \cite{shaw2009alchemy} is given as follows:

Let $X_1$ and $X_2$ be independent and identically distributed (iid) random variables with cumulative distribution function (CDF) $G\left(.\right)$ and probability density function (PDF) $g\left(.\right)$, and $X_{1:n} $, $X_{2:n} $ be the first two order statistics associated with the sample.

 Let us define random variable $T$ 
 \begin{equation*}
 \begin{array}{l}
 T\overset{d}{=}X_{1:2},\mbox{ with probability }\omega,  \\ 
 T\overset{d}{=}X_{2:2},\mbox{ with probability }1-\omega ,%
 \end{array}%
 \end{equation*}
 and corresponding CDF is
\begin{eqnarray}
\label{transmuted: cdf}
F_{T}\left( x\right)  &=&\omega P\left( {X_{1:2}\leq x}\right) +\left( {1-\omega }%
\right) P\left( {X_{2:2}\leq x}\right)   \notag \\
&=&\omega \left( {1-\left( {1- G\left( x\right)}\right)^{2} }%
\right) +\left( {1-\omega }\right) G^{2}\left( x\right)   \notag \\
&=&2\omega G\left( x\right) +\left( {1-2\omega }\right) G^{2}\left( x\right) .
\end{eqnarray}
 where $\omega \in (0,1) $

Substituting $\omega =\frac{1+\lambda }{2}$ in Eq. (\ref{transmuted: cdf}), the CDF can be rewritten as follows: 

\begin{equation}
\label{eq2}
F_{T}\left(x\right)=\left(1+\lambda\right)G\left(x\right)-\lambda \left(G\left(x\right)\right)^2,
\end{equation}
and the corresponding PDF is
\begin{equation}
\label{eq3}
f_{T}\left(x\right)=\left(1+\lambda\right)g\left(x\right)-2\lambda G\left(x\right)g\left(x\right),
\end{equation}
where $\lambda \in \left[-1,1\right]$.
In last decades, the transmutation map based on the distributions of the order statistics has been used by many authors to generate the flexible statistical distributions. Some of these studies are listed as follows: 
Aryal and Tsokos \cite{aryal2011transmuted} proposed transmuted Weibull distribution as an alternative to Weibull distribution. Mahmoud and Mandouh \cite{mahmoud2013transmuted} studied on transmuted Fr{\'e}chet distribution Transmuted inverse Rayleigh distribution is suggested by Afaq et al. \cite{ahmad2014transmuted}. Tan{\i}{\c{s}} et al. \cite{tanics2020transmuted} provided transmuted complementary exponential power distribution.


Then, Granzotto et al. \cite{granzotto2017cubic} proposed the cubic rank transmuted family of distribution using the cubic rank  transmutation map (CRTM) based on the distributions of first three order statistics. The CRTM can be summarized as follows:

Let $X_1 $,$X_2 $,$X_3 $ be iid random variables with the CDF $G\left(x\right)$ and the PDF $g\left(x\right)$. $X_{1:n} $,$X_{2:n} $ and $X_{3:n} $ be the order statistics of this sample.
 Let us define a random variable $Z$ 
\[
\begin{array}{l}
 Z\overset{d}{=} \mbox{ X}_{1:3} ,\mbox{ with probability 
}p_1, \\ 
Z\overset{d}{=} \mbox{ X}_{2:3} ,\mbox{ with probability 
}p_2, \\ 
 Z\overset{d}{=} \mbox{ X}_{3:3} ,\mbox{ with probability 
}p_3, \\ 
 \end{array}
\]

where $p_1 +p_2 +p_3 =1$. The corresponding CDF of $Z\mbox{ }$ is
\begin{equation}
 \label{cubic}
\begin{array}{l}
 F_Z \left( x \right)=p_1 P\left( {X_{1:3} \le x} \right)+p_2 P\left( 
{X_{2:3} \le x} \right)+p_3 P\left( {X_{3:3} \le x} \right) \\ 
 \\ 
 {\mathrm{\;\;\;\;\;\;\;\;}}\mbox{ }=p_1 \left[ {1-\left( {1-G\left( x \right)} \right)^3} \right]+6p_2 
\int\limits_0^x {G\left( t \right)\left( {1-G\left( t \right)} 
\right)g\left( t \right)dt} +p_3 G^3\left( x \right) \\ 
 \\ 
{\mathrm{\;\;\;\;\;\;\;\;}} \mbox{ }=3p_1 G\left( x \right)+\left( {3p_2 -3p_1 } \right)G^2\left( x 
\right)+\left( {1-3p_2 } \right)G^3\left( x \right). \\ 
 \end{array}
\end{equation}
Substituting $3p_{1}=\lambda _{1},3p_{2}=\lambda _{2}$ in Eq. (\ref{cubic}),
the CDF and PDF of $Z$ are
\begin{equation}
F_{Z}\left( x\right) =\lambda _{1}G\left( x\right) +\left( \lambda
_{2}-\lambda _{1}\right) G^{2}\left( x\right) +\left( 1-\lambda _{2}\right)
G^{3}\left( x\right) ,  \label{eq5}
\end{equation}%
and 
\begin{equation}
f_{Z}\left( x\right) =g\left( x\right) \left[ \lambda _{1}+2\left( \lambda
_{2}-\lambda _{1}\right) G\left( x\right) +3\left( 1-\lambda _{2}\right)
G^{2}\left( x\right) \right] ,  \label{eq6}
\end{equation}%
respectively, where $\lambda _{1}\in \left[ 0,1\right] ,\lambda _{2}\in %
\left[ -1,1\right] $.
\citet{granzotto2017cubic} proposed two sub models of the cubic rank transmuted family of distributions (CRTFD) based on Weibull and log-logistic distributions. 
Then, Sara{\c{c}}o{\u{g}}lu and Tan{\i}{\c{s}} \cite{saraccouglu2018new} introduced a new member of the CRTFD based on Kumaraswamy distribution.
Also, the cubic rank transmuted Kumaraswamy distribution stands out as the first member of the CRTFD whose random variable is defined on the unit interval (0,1). Tan{\i}{\c{s}} and Sara{\c{c}}o{\u{g}}lu \cite{tanics2022cubic} studied on cubic rank transmuted inverse Rayleigh distribution. Tan{\i}{\c{s}} and Sara{\c{c}}o{\u{g}}lu \cite{tanics2023cubic} introduced a new member of the CRTFD based on generalized Gompertz distribution.

Likewise, the CRTFD as presented in Eq. (\ref{transmuted: cdf}), Balakrishnan and He \cite{balakrishnan2021record} suggested a new family of distributions based on the first two upper record statistics, namely record-based transmuted family of distributions (RBTFD). The RBTFD is obtained by record-based transmutation map (RBTM) defined as follows:

Let $X_1 $ and $X_2 $  be a random sample with two sizes from the distribution with CDF $G(.)$ and PDF $g(.)$, and $X_{U\left( 1 \right)} $ and $X_{U\left( 2 \right)} $ 
be upper records of the corresponding sample.
 
 Let us define a random variable $R$ 
\[
\begin{array}{l}
 R\overset{d}{=} \mbox{ }X_{U\left( 1 \right)} ,\mbox{ 
with probability }p _1, \\ 
 R\overset{d}{=} \mbox{ }X_{U\left( 2 \right)} ,\mbox{ 
with probability }p _2, \\ 
 \end{array}
\]
where $U_{\left( n\right) }=\min \left\{ i:i>U\left( n-1\right)
,X_{i}>X_{U\left( n-1\right) }\right\} \left\{ U_{\left( n\right) }\right\}
_{n=1}^{\infty }$ denotes upper record times and $\left\{ X_{U\left(
n\right) }\right\} _{n=1}^{\infty }$ refers to the corresponding record sequence \cite{balakrishnan2021record},\citet{arnold2008first}, $p_{1}+p_{2}=1$ and $R$ denotes a random variable from  record-based transmuted distribution. Thus, the CDF of $R$ is 
\begin{eqnarray}
\label{rbt: cdf}
F_{R}\left( x\right)  &=&p_{1}P\left( {X_{U\left( 1\right) }\leq x}\right)
+p_{2}P\left( {X_{U\left( 2\right) }\leq x}\right)   \nonumber \\
&=&G\left( x\right) +p\left[ {\left( {1-G\left( x\right) }\right) \ln
\left( {1-G\left( x\right) }\right) }\right], 
\end{eqnarray}
and the corresponding PDF is  
 
\begin{equation}
\label{rbt: pdf}
f_R \left( x \right)=g\left( x \right)\left[ {1+p\left( {-\ln \left( 
{1-G\left( x \right)} \right)-1} \right)} \right],
\end{equation}
where $p\in \left( {0,1} \right)$. Six sub-models of the RBTFD were introduced based on uniform, exponential, linear exponential, Weibull, Normal, logistic distributions in \cite{balakrishnan2021record}. However, no statistical inferences were provided regarding the proposed distributions in \cite{balakrishnan2021record}. Tan{\i}{\c{s}} and Sara{\c{c}}o{\u{g}}lu \cite{tanics2022record} provided the statistical inferences and theoretical properties of the record-based transmuted Weibull and exponential distributions.
In the literature, there are a few special cases of RBTFD based on Lindley \cite{tanics2024new}, generalized linear exponential \cite{arshad2024record}, power Lomax \cite{sakthivel2022record}, and Burr X \cite{alrweili2025statistical}.



Also, the cubic record-based transmuted family of distributions (CRBTFD), based on the distributions of the first three upper record values, was introduced by \citet{balakrishnan2021record}. The CRBTFD is derived using the cubic record-based transmutation map (CRBTM). The CRBTM is given as follows: 
 
Let $X_1 $, $X_2$, and $X_3$  be iid random variables with the CDF $G(.)$ and PDF $g(.)$, and $X_{U\left( 1 \right)}$, $X_{U\left( 2 \right)} $, and $X_{U\left( 3 \right)} $ be upper records associated the sample.
 
 Let us define a random variable $R^{*}$ 
\[
\begin{array}{l}
 R^{*}\overset{d}{=} \mbox{ }X_{U\left( 1 \right)} ,\mbox{ 
with probability }p _1, \\ 
 R^{*}\overset{d}{=} \mbox{ }X_{U\left( 2 \right)} ,\mbox{ 
with probability }p _2, \\ 
R^{*}\overset{d}{=} \mbox{ }X_{U\left( 3 \right)} ,\mbox{ 
with probability }1-p _1-p _2, \\
 \end{array}
\]
where $R^{*}$ denotes a random variable from a distribution in CRBTFD and the CDF of $R^{*}$ is 
\begin{eqnarray}
\label{crbt: cdf}
F_{R^{*}}\left( x\right)  &=&p_{1}P\left( {X_{U\left( 1\right) }\leq x}\right)
+p_{2}P\left( {X_{U\left( 2\right) }\leq x}\right) +\left(
1-p_{1}-p_{2}\right) P\left( {X_{U\left( 3\right) }\leq x}\right)   \notag \\
&=&1-\left( 1-G\left( x\right) \right) \left( 1+\left( 1-p_{1}\right) \left(
-{\ln \left( {1-G\left( x\right) }\right) }\right) \right)   \notag \\
&&+\left( 1-G\left( x\right) \right) \left( \frac{1-p_{1}-p_{2}}{2}\left( -{%
\ln \left( {1-G\left( x\right) }\right) }\right) ^{2}\right), 
\end{eqnarray}%
where $\ 0<p_{1},p_{2}<1,p_{1}+p_{2}\leq 1$. The corresponding PDF of $R^{*}$ is  
 
\begin{equation}
\label{crbt: pdf}
f_{R^{*}}\left( x\right) =g\left( x\right) \left( p_{1}+p_{2}\left( -{\ln
\left( {1-G\left( x\right) }\right) }\right) {+}\frac{1-p_{1}-p_{2}}{2}%
\left( -{\ln \left( {1-G\left( x\right) }\right) }\right) ^{2}\right) .
\end{equation}
Tan{\i}{\c{s}} \cite{tanics2025cubic} provide statistical inferences about CRBTFD and a special case based on the Weibull distribution.
In addition to RBTM, Balakrishnan and He \cite{balakrishnan2021record} suggested a dual record-based transmutation map (DRBTM) to generate flexible distributions from the dual record-based transmuted family of distributions (DRBTFD). The DRBTM is based on distributions of the first two lower record values and is defined as follows:   


Let $X_1 $ and $X_2 $  be iid random variables with CDF $G(.)$ and PDF $g(.)$, and $X_{L\left( 1 \right)} $ and $X_{L\left( 2 \right)} $ 
be lower records of corresponding sample.
 
 Let us define a random variable $L$ 
\[
\begin{array}{l}
 L\overset{d}{=} \mbox{ }X_{L\left( 1 \right)} ,\mbox{ 
with probability }1-p \\ 
 L\overset{d}{=} \mbox{ }X_{L\left( 2 \right)} ,\mbox{ 
with probability }p, \\ 
 \end{array}
\]
where $L_{\left( n\right) }=\min \left\{ i:i>L\left( n-1\right)
,X_{i}<X_{L\left( n-1\right) }\right\} \left\{ L_{\left( n\right) }\right\}
_{n=1}^{\infty }$ refers to lower record times and $\left\{ X_{L\left(
n\right) }\right\} _{n=1}^{\infty }$ denotes the corresponding record sequence \cite{balakrishnan2021record},\citet{arnold2008first}, $L$ refers to a random variable distributed any member of DRBTFD. Thus, the CDF of $L$ is 
\begin{eqnarray}
	F_{L}\left( x\right) &=&\left(
		1-p\right)P\left( X_{L\left( 1\right) }\leq x\right) +p P\left( X_{L\left( 2\right) }\leq x\right)  \notag \\
	&=&\left(
		1-p\right)G\left( x\right) +p \left[ G\left( x\right) \left(
	1-\ln \left( G\left( x\right) \right) \right) \right]  \notag \\
	&=&G\left( x\right) \left[ 1-p\ln \left( G\left( x\right) \right) \right],
	\label{trltcdf}
	\end{eqnarray}%
 and the corresponding PDF is  
 
\begin{equation}
\label{eq9}
f_{L} \left( x \right)=g\left( x \right)\left[ {1+p\left( {-\ln \left( 
{1-G\left( x \right)} \right)-1} \right)} \right].
\end{equation}
where $p\in \left( {0,1} \right)$. In recent years, some studies about the DRBTFD were considered based on Fr{\'e}chet \cite{tanics2021frechet}, power function \cite{tanis2021transmuted} and inverse Rayleigh \cite{tanis2022record} distributions.

Balakrishnan and He \cite{balakrishnan2021record} have briefly noted the order-3 version of the DRBTFD in their paper, though no detailed exposition was provided. Inspired by this work, the main objective of the present study is to propose the cubic lower record-based transmutation map (CLRBTM) and to provide a foundation for future studies aimed at generating new distributions within the cubic lower record-based transmuted family of distributions (CLRBTFD), using the distributions of first three lower record values.

The present work is organized as follows: In Section \ref{sec: clrbtm}, we suggest the CLRBTM to generate continuous distribution. In addition, a special case of CLRBTFD based on exponential distribution is proposed in Section \ref{sec: clrbtm}. Section \ref{properties} provides some distributional properties of the suggested distribution. In Section \ref{Estimation methods}, we focus on the point estimation of the introduced distribution via nine different methods. Then, we design a comprehensive Monte Carlo (MC) simulation study to compare the performance of these estimators in Section \ref{simulation}. In Section \ref{sec:real}, two real-life data examples are presented to illustrate the usefulness of the proposed distribution in real-life data modeling. Lastly, concluding remarks are given in Section \ref{conc}. 

\section{Cubic lower record-based transmuted family of distributions}
\label{sec: clrbtm}
Balakrishnan and He \cite{balakrishnan2021record} mentioned the CLRBTM based on distributions of the first three lower record values. However, no explicit formulation of CLRBTM was given in \cite{balakrishnan2021record}. In this section, we present both the construction of the CLRBTM and a detailed account of the resulting CLRBTFD obtained through the CLRBTM defined as follows:

Let $X_1 $, $X_2$, and $X_3$  be a random sample with CDF $G(.)$ and PDF $g(.)$, and $X_{L\left( 1 \right)}$, $X_{L\left( 2 \right)} $, and $X_{L\left( 3 \right)} $ be lower records of same sample.
 
 Let us define a random variable $L^{*}$ 
\[
\begin{array}{l}
 L^{*}\overset{d}{=} \mbox{ }X_{L\left( 1 \right)} ,\mbox{ 
with probability }p _1, \\ 
 L^{*}\overset{d}{=} \mbox{ }X_{L\left( 2 \right)} ,\mbox{ 
with probability }p _2, \\ 
L^{*}\overset{d}{=} \mbox{ }X_{L\left( 3 \right)} ,\mbox{ 
with probability }1-p _1-p _2, \\
 \end{array}
\]
where $p_{1}+p_{2}=1$ and $L^{*}$ denotes a random variable from a distribution of CLRBTFD and the CDF of $L^{*}$ is 
\begin{eqnarray}
\label{clrbt: cdf}
F_{L^{*}}\left( x\right)  &=&p_{1}P\left( {X_{L\left( 1\right) }\leq x}\right)
+p_{2}P\left( {X_{L\left( 2\right) }\leq x}\right) +\left(
1-p_{1}-p_{2}\right) P\left( {X_{L\left( 3\right) }\leq x}\right)   \notag \\
&=&p_{1}G\left( x\right) +p_{2} \left[ G\left( x\right) \left(
	1-\ln \left( G\left( x\right) \right) \right) \right]+(1-p_{1}-p_{2})[G(x)(1-\ln G(x)+\frac{1}{2}(-\ln G(x))^2)] \notag \\
    &=&G(x)\left[1-(1-p_{1})\ln G(x)+\frac{1-p_{1}-p_{2}}{2}\big(-\ln G(x)\big)^{2}\right], 
\end{eqnarray}%
where $\ 0<p_{1},p_{2}<1,p_{1}+p_{2}\leq 1$. Thus, the corresponding PDF of $L^{*}$ is  
 
\begin{equation}
\label{clrbt: pdf}
f_{L^{*}}\left( x\right) =g(x)\Big[ p_{1} - p_{2}\ln G(x) + \frac{1-p_{1}-p_{2}}{2}\big(\ln G(x)\big)^{2}\Big] .
\end{equation}


\subsection{Cubic Lower Record-Based Transmuted Exponential Distribution}

In this section, we propose the cubic lower record-based transmuted exponential distribution (CLRBTE) using the CLRBTM.

Let \(X \sim \text{CLRBTE}(\lambda, p_1, p_2)\) denote a random variable.
By substituting the exponential CDF \(G(x) = 1 - e^{-\lambda x}\) into Eq. (\ref{clrbt: cdf}), the CDF and PDF of \(X\) are obtained as follows:
\begin{equation}
\label{clrbte: cdf}
F_{CLRBTE}(x) = \big(1 - e^{-\lambda x}\big) \left[ 1 - (1-p_1) \ln\big(1 - e^{-\lambda x}\big) + \frac{1-p_1-p_2}{2} \big(\ln(1 - e^{-\lambda x})\big)^2 \right].
\end{equation}
and

\begin{equation}
\label{clrbte: pdf}
f_{CLRBTE}(x) = \lambda e^{-\lambda x} \Big[ p_1 - p_2 \ln(1 - e^{-\lambda x}) + \frac{1-p_1-p_2}{2} \big(\ln(1 - e^{-\lambda x})\big)^2 \Big], \quad x \ge 0,
\end{equation}
respectively, where \(p_1, p_2 \in [0,1]\) and \(p_1 + p_2 \le 1\). The CLRBTE distribution reduces to the standard exponential distribution when \(p_1 = 1\) and \(p_2 = 0\).

\begin{figure} [H]
    \centering
    \includegraphics[width=0.6\linewidth]{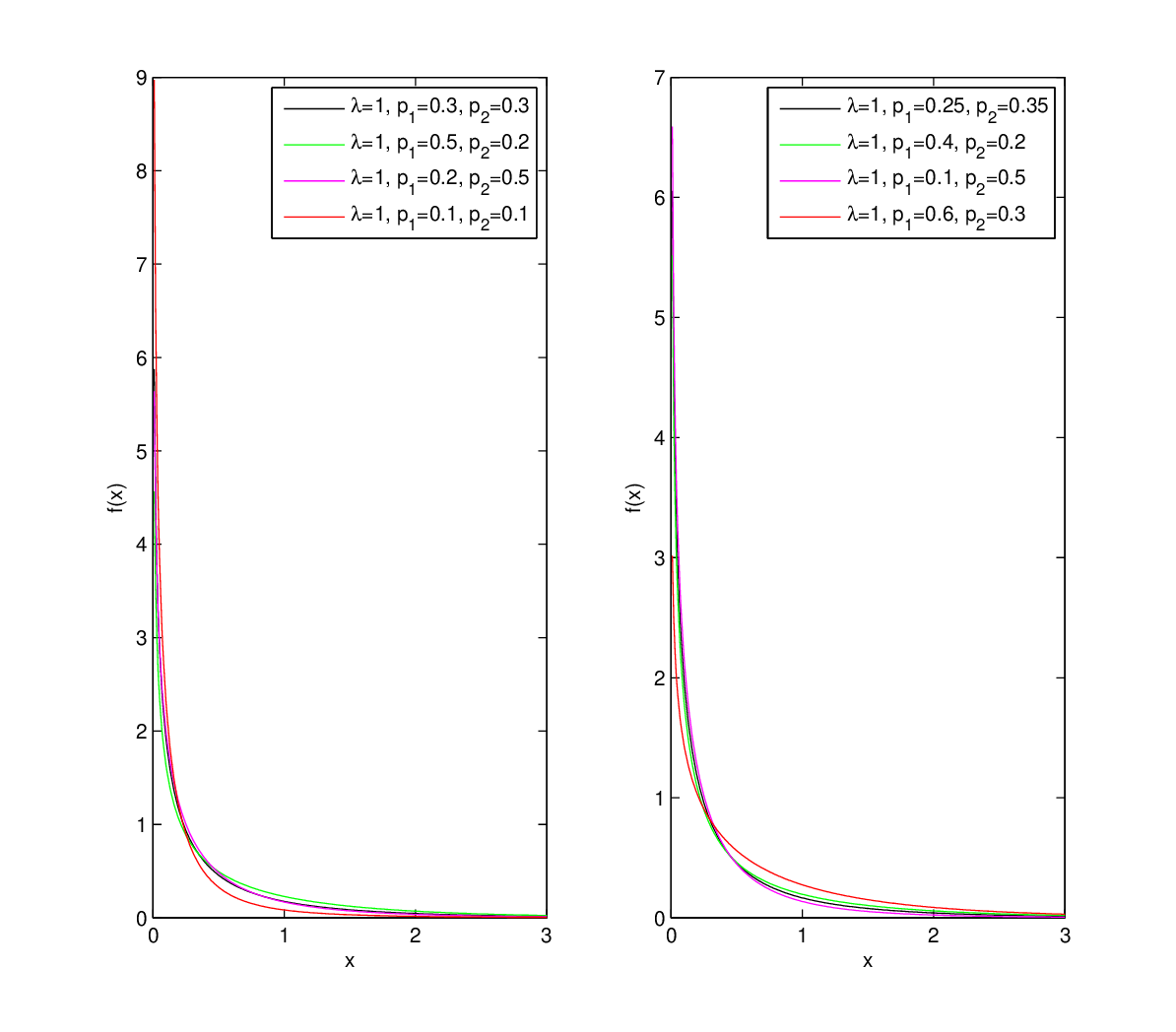}
    \caption{The PDFs of CLRBTE distribution under the different values of parameters}
\label{pdf}
\end{figure}

Figure~\ref{pdf} illustrates the possible PDF shapes associated with the CLRBTE distribution in the parameters selected values. From Figure \ref{pdf}, we clearly see that the CLRBTE distribution has a decreasing PDF.

\section{Some distributional properties}
\label{properties}
This section explores some statistical properties of the CLRBTE distribution, namely, hazard function (hf), moments, mean, variance, coefficient of variation, coefficient of skewness, and coefficient of kurtosis.





 \subsection{Hazard function}
The hf of the CLRBTE distribution is 
\begin{align}
h(x) &= \frac{f(x)}{1 - F(x)} \notag \\
&= 
\frac{\lambda e^{-\lambda x} \left[ p_1 - p_2 \ln(1 - e^{-\lambda x}) + \frac{1-p_1-p_2}{2} (\ln(1 - e^{-\lambda x}))^2 \right]}
{1 - (1 - e^{-\lambda x}) \left[ 1 - (1-p_1) \ln(1 - e^{-\lambda x}) + \frac{1-p_1-p_2}{2} (\ln(1 - e^{-\lambda x}))^2 \right]}.
\end{align}
where $F(.)$ and $f(.)$ are the CDF and PDF of CLRBTE distribution respectively given in Eqs. (\ref{clrbte: cdf}) and (\ref{clrbte: pdf})
Figure \ref{hf} illustrates the possible hazard shapes of the CLRBTE distribution.

\begin{figure} [H]
    \centering
    \includegraphics[width=0.6\linewidth]{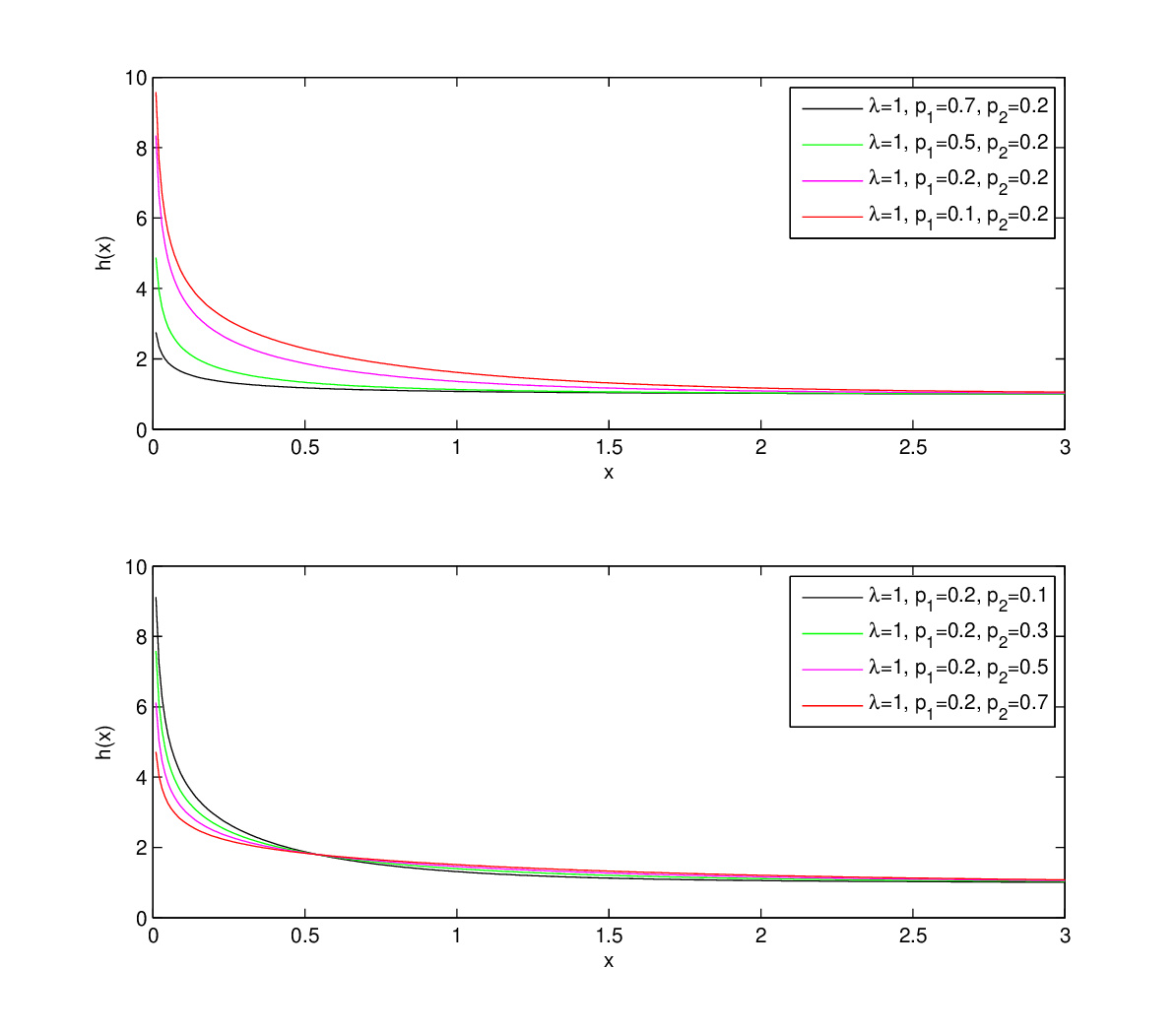}
    \caption{The hf shapes of the CLRBTE distribution in the different values of parameters}
\label{hf}
\end{figure}
Figure \ref{hf} displays that the hf is decreasing shape.

\subsection{Moments}
In this subsection, we present the $r$-th central moment for the $CLRBTE(\lambda,p_{1},p_{2})$ distribution via the Theorem \ref{teo:1}.
\begin{theorem}
\label{teo:1}
Let $X$ be a random variable from the $CLRBTE(\lambda,p_{1},p_{2})$ distribution. The $r$-th central moment is

\begin{align}
\label{moment}
\mathbb{E}[X^r] =& \int_0^{\infty} x^r f_{CLRBTE}(x) \, dx \\\nonumber
=&x^r\lambda e^{-\lambda x} \Big[ p_1 - p_2 \ln(1 - e^{-\lambda x}) + \frac{1-p_1-p_2}{2} \big(\ln(1 - e^{-\lambda x})\big)^2 \Big] \\\nonumber
=& \frac{1}{\lambda^r} \sum_{k_1=1}^{\infty} \cdots \sum_{k_r=1}^{\infty} \frac{1}{k_1 \cdots k_r} \left[ \frac{p_1}{K+1} + \frac{p_2}{(K+1)^2} + \frac{1-p_1-p_2}{(K+1)^3} \right], \quad K = k_1 + \dots + k_r.
\end{align}
\end{theorem}

\begin{proof}
To obtain the r-th moment, we apply the substitution \(u = 1 - e^{-\lambda x}\), which implies \(x = -\frac{1}{\lambda}\ln(1-u)\) and \(dx = \frac{du}{\lambda(1-u)}\). Under this transformation, the integral becomes
\[
\mathbb{E}[X^r] = \frac{1}{\lambda^r} \int_0^1 \big(-\ln(1-u)\big)^r \Big[ p_1 - p_2 \ln u + \frac{1-p_1-p_2}{2} (\ln u)^2 \Big] du.
\]

Expanding \(-\ln(1-u)\) as a power series
\[
-\ln(1-u) = \sum_{k=1}^{\infty} \frac{u^k}{k}, \quad |u|<1,
\]
and interchanging summation and integration (justified by absolute convergence), we obtain
\[
\mathbb{E}[X^r] = \frac{1}{\lambda^r} \sum_{k_1=1}^{\infty} \cdots \sum_{k_r=1}^{\infty} \frac{1}{k_1 \cdots k_r} \int_0^1 u^{k_1+\dots+k_r} \Big[ p_1 - p_2 \ln u + \frac{1-p_1-p_2}{2} (\ln u)^2 \Big] du.
\]

Using the standard integrals
\[
\int_0^1 u^m du = \frac{1}{m+1}, \quad
\int_0^1 u^m \ln u \, du = -\frac{1}{(m+1)^2}, \quad
\int_0^1 u^m (\ln u)^2 \, du = \frac{2}{(m+1)^3},
\]
with \(m = k_1 + \dots + k_r\), we have
\[
\int_0^1 u^m \Big[ p_1 - p_2 \ln u + \frac{1-p_1-p_2}{2} (\ln u)^2 \Big] du
= \frac{p_1}{m+1} + \frac{p_2}{(m+1)^2} + \frac{1-p_1-p_2}{(m+1)^3}.
\]
Therefore, the moment r-th can be expressed as a convergent series
\[
\mathbb{E}[X^r] = \frac{1}{\lambda^r} \sum_{k_1=1}^{\infty} \cdots \sum_{k_r=1}^{\infty} \frac{1}{k_1 \cdots k_r} \left[ \frac{p_1}{K+1} + \frac{p_2}{(K+1)^2} + \frac{1-p_1-p_2}{(K+1)^3} \right], \quad K = k_1 + \dots + k_r.
\]
Therefore, the proof is finished.
\end{proof}

\subsection{Coefficients of Variation, Skewness and Kurtosis}
In this subsection, we introduce the Coefficient of skewness $\gamma_1$ and coefficient of kurtosis $\gamma_2$ regarding CLRBTE $\left( \lambda,p_{1},p_{2}\right) $ distribution are 
given in Eqs. (\ref{cs}) and (\ref{ck}).
\begin{equation}
\label{cs}
\gamma_1=\frac{\mu'_3 - 3\mu'_2\mu'_1 + 2(\mu'_1)^3}{(\sigma)^3}
\end{equation}
and
\begin{equation}
\label{ck}
\gamma_2=\frac{\mu'_4 - 4\mu'_3\mu + 6\mu'_2\mu^2 - 3\mu^4}{(\sigma)^4},
\end{equation}
respectively, where first four moments $\mu'_i$, $i=1,2,3,4$ can be obtained by substituting $r=1,2,3,4$ in Eq. (\ref{moment}), $\mu'_1$ is mean,  and $\sigma^2$ is variance, \( \sigma^2 = \mu'_2 - (\mu'_1)^2 \).

Also, the variation coefficient is
\begin{equation}
\text{CV} = \frac{\sigma}{\mu}.
\end{equation}

\begin{table}[H]
\centering
\caption{Some descriptive statistics in various values of parameters}
\label{tab:parameter_effect}
\begin{tabular}{ccc|cccccc}
\hline
$\lambda$ & $p_1$ &  $p_2$ &  $\mu'_1$ & $\sigma^2$ & $\sigma$ & CV & $\gamma_1$ & $\gamma_2$ \\
\hline
0.5 & 0.5  & 0.5 & 1.3550 & 2.7758  & 1.666 & 1.2295 & 2.5782 & 12.9647 \\
1 & 0.5  & 0.5  & 0.6775 & 0.6939 & 0.8330 & 1.2295 & 2.5782 & 12.9647 \\
2 & 0.5  & 0.5  & 0.3387 & 0.1734 & 0.4165 & 1.2295 & 2.5782 & 12.9647 \\
1.5 & 0.1  & 0.6  & 0.2392 & 0.1226 & 0.3502 & 1.4637 & 3.6422 & 25.1704 \\
1.5 & 0.3  & 0.6  & 0.3522 & 0.2273 & 0.4768 & 1.3537 & 3.0116 & 16.9262  \\
1.5 & 0.6  & 0.6  & 0.5216 & 0.3365 & 0.5801 & 1.1121 & 2.4077 & 11.7738 \\
1.5 & 0.9  & 0.6  & 0.6910 & 0.3883 & 0.6231 & 0.9018 & 2.1977 & 10.2468 \\
0.7 & 0.4  & 0.1  & 0.7314 & 1.2324 & 1.1101 & 1.5177 & 3.0031 & 15.9308 \\
0.7 & 0.4  & 0.4  & 0.8180 & 1.2422 & 1.1145 & 1.3624 & 2.8380 & 14.9678 \\
0.7 & 0.4  & 0.7  & 0.9046 & 1.2369 & 1.1122 & 1.2294 & 2.7246 & 14.3841 \\
0.7 & 0.4  & 0.9  & 0.9623 & 1.2251 & 1.1068 & 1.1501 & 2.6775 & 14.1848 \\

\hline
\end{tabular}
\end{table}

Table \ref{tab:parameter_effect} provides the descriptive statistics of the CLRBTE distribution for selected combinations of the parameters $\lambda$, $p_{1}$, and $p_{2}$. The scale parameter $\lambda$ primarily regulates the overall magnitude of the distribution, as reflected in the consistent decrease of both mean and variance with higher values of $\lambda$, whereas shape-related measures such as the coefficient of variation, skewness, and kurtosis remain unchanged, indicating that scaling influences dispersion but leaves the distributional form intact.
 In contrast, the transmutation parameters $p_{1}$ and $p_{2}$ play a pivotal role in the determination of higher-order characteristics. Lower values of $p_{1}$, particularly when combined with moderate or large $p_{2}$, accentuate variability, asymmetry, and tail heaviness, leading to more leptokurtic and skewed profiles. In contrast, larger values of $p_{1}$ temper these effects, producing distributions that are closer to symmetry with lighter tails. Considered together, these findings underscore a dual parameterization mechanism: $\lambda$ governs the overall scale, while $p_{1}$ and $p_{2}$ impart substantial flexibility in shaping the distribution, thus improving the adaptability of the CLRBTE distribution for various empirical applications. 


\section{Point Estimation}\label{Estimation methods}
In this section, we tackle the point estimation of the CLRBTE distribution. Nine methods are utilized to estimate the parameters $\lambda$, $p_1$, and $p_2$. 

\subsection{Maximum likelihood Estimator}
This subsection explores  ML estimators (MLEs) of the parameters $\lambda$, $p_1$, and $p_2$. 

Let ${X_1},{X_2}...{X_n}$ be iid random variables from the CLRBTE and ${x_1},{x_2}...{x_n}$  denote the observed values of the sample. Then, the log-likelihood function is as follows:

\begin{align}
\label{mle}
\ell(x; \lambda,p_1,p_2) = n \ln \lambda - \lambda \sum_{i=1}^n x_i 
+ \sum_{i=1}^n \ln \left( 
p_1 - p_2 \ln\!\big(1-e^{-\lambda x_i}\big) 
+ \frac{1-p_1-p_2}{2} \left[\ln\!\big(1-e^{-\lambda x_i}\big)\right]^2
\right).
\end{align}
The first-order partial derivatives of $\ell(x; \lambda,p_1,p_2)$ over the parameters $\lambda$, $p_1$, and $p_2$ are given by

\begin{align}
\label{der_lambda}
\frac{\partial \ell(x; \lambda,p_1,p_2)}{\partial \lambda}
&= \frac{n}{\lambda} - \sum_{i=1}^n x_i \notag \\
&\quad+ \sum_{i=1}^n
\frac{\,x_i e^{-\lambda x_i}\,
\Big[-p_2 + (1-p_1-p_2)\ln\!\big(1-e^{-\lambda x_i}\big)\Big]}
{(1-e^{-\lambda x_i})\Big(p_1 - p_2\ln\!\big(1-e^{-\lambda x_i}\big)
+ \tfrac{1-p_1-p_2}{2}\,[\ln(1-e^{-\lambda x_i})]^2\Big)}.
\end{align}

\begin{align}
\label{der_p1}
\frac{\partial \ell(x; \lambda,p_1,p_2)}{\partial p_1}
&= \sum_{i=1}^n
\frac{\,1 - \tfrac{1}{2}\,[\ln(1-e^{-\lambda x_i})]^2\,}
{\,p_1 - p_2\ln(1-e^{-\lambda x_i})
+ \tfrac{1-p_1-p_2}{2}\,[\ln(1-e^{-\lambda x_i})]^2\,}.
\end{align}

\begin{align}
\label{der_p2}
\frac{\partial \ell(x; \lambda,p_1,p_2)}{\partial p_2}
&= \sum_{i=1}^n
\frac{-\ln(1-e^{-\lambda x_i}) - \tfrac{1}{2}\,[\ln(1-e^{-\lambda x_i})]^2}
{\,p_1 - p_2\ln(1-e^{-\lambda x_i})
+ \tfrac{1-p_1-p_2}{2}\,[\ln(1-e^{-\lambda x_i})]^2\,},
\end{align}

respectively. The MLEs of the parameters $\lambda$, $p_1$, and $p_2$ are  values of parameters that maximize $\ell(x; \lambda,p_1,p_2)$. The MLEs are simultaneously solution the Equations obtained by setting each of the Eqs. (\ref{der_lambda}), (\ref{der_p1}), and (\ref{der_p2}) equal to zero.

Some numerical methods, namely, Newton-Raphson and Nelder Mead can be utilized to solve this optimization problem
\subsection{Least squares estimator}
In this subsection, we examine the LS estimators (LSEs) the parameters $\lambda$, $p_1$, and $p_2$. The LSE is proposed by \cite{swain1988least} as an alternative to the MLE. The LSEs can be obtained by minimizing the function given in Eq. (\ref{lse}).

\begin{align}
\label{lse}
LS(x_{i})&=\sum_{i=1}^{n}\left[F(x_{i:n})-\frac{i}{n+1}\right]^2,
\end{align}
where ${x_{i:n}}$ for $i = 1,2...n$ denote the order statistics.
\subsection{Weighted least squares estimator}
This subsection introduces the WLS estimators (WLSEs) of parameters $\lambda$, $p_1$, and $p_2$ for the CLRBTE via the method proposed by \cite{swain1988least}. We derive the WLSEs by minimizing the Eq. (\ref{wlse}).
\begin{align}
\label{wlse}
WLS(x_{i})&=\sum_{i=1}^{n}\frac{(n+1)^2(n+2)}{i(n-i+1)}\left[F(x_{i:n})-\frac{i}{n+1}\right]^2.
\end{align}

\subsection{Anderson-Darling estimator}
This subsection discusses the AD estimators (ADEs) of parameters $\lambda$, $p_1$, and $p_2$. The ADEs of the parameters $\lambda$, $p_1$, and $p_2$ can be derived by minimizing Eq (\ref{ade}).

\begin{align}
\label{ade}
AD(x_{i})&=-n-\frac{1}{n}\sum_{i=1}^{n}(2i-1)\left[\ln F(x_{i:n})+\ln (1-F(x_{n-i-1:n}))\right].
\end{align}


\subsection{Cram\'{e}r-von Mises estimator}
In this subsection, we study the CvM estimators (CvMEs) of the parameters $\lambda$, $p_1$, and $p_2$. The Cram\'{e}r-von Mises method is suggested by \cite{choi1968estimation}. We compute the CvMEs by maximizing the function in Eq. (\ref{cvme}).

\begin{align}
\label{cvme}
C(x_{i})&=\frac{1}{12n}+\sum_{i=1}^{n}\left[F(x_{i:n})-\frac{2i-1}{2n}\right]^2.
\end{align}

\subsection{Maximum product of spacings estimator}
This subsection introduces the MPS method proposed by \cite{cheng1983estimating} and \cite{ranneby1984maximum} as an alternative to the ML method, the MPS method based on maximizing the function in Eq. (\ref{mpse}) to obtain the MPS estimators (MPSEs). 
\begin{align}
\label{mpse}
\delta\left(x_{i} \right) =\frac{1}{n+1}\sum_{i=1}^{n+1}\ln I_{i}(x_{i}),
\end{align}

where $I_{i}(x_{i})=F(x_{i:n})-F(x_{i-1:n})$, $F(x_{0:n})=0$ and $F(x_{n+1:n})=1.$

\subsection{Right tail Anderson Darling estimator}
In this subsection, we explore the RTAD estimators (RTADEs) of parameters $\lambda$, $p_1$, and $p_2$. We derive the RTADEs by minimizing the Eq. (\ref{RTADE}).

\begin{equation} \label{RTADE}
\begin{aligned}
L(\lambda,p_1,p_2) 
&= \frac{n}{2} - 2\sum\limits_{i = 1}^n {F\left( {{x_{i:n}}} \right) - \frac{1}{n}} \sum\limits_{i = 1}^n {\left( {2i - 1} \right)\ln \left( {1 - F\left( {{x_{n - i + 1}}} \right)} \right)} .
\end{aligned}
\end{equation}

\subsection{Minimum spacing absolute distance estimator}
In this subsection, we suggest the MSAD estimators (MSADEs) of the parameters $\lambda$, $p_1$, and $p_2$. The MSADEs of parameters $\lambda$, $p_1$, and $p_2$ are computed by minimizing the function given in Eq. (\ref{msad}).

\begin{equation} \label{msad}
\begin{aligned}
\Lambda (\lambda,p_1,p_2 \mid x_i) 
= \sum\limits_{i = 1}^{n + 1} {\left| {{I_{i}} - \frac{1}{{n + 1}}} \right|}
\end{aligned}
\end{equation}

\subsection{Minimum spacing absolute-log distance estimator}
In this subsection, we tackle the MSALD estimators (MSALDEs) of parameters $\lambda$, $p_1$, and $p_2$. \cite{torabi2008general} suggested the MSALD method as an alternative the ML. We derive the MSALDEs of the parameters $\lambda$, $p_1$, and $p_2$ by minimizing the Eq. (\ref{MSALD})

\begin{equation} \label{MSALD}
\begin{aligned}
\Psi (\lambda,p_1,p_2 \mid x_i) 
= \sum\limits_{i = 1}^{n + 1} {\left| {{I_i} - \log \frac{1}{{n + 1}}} \right|}.
\end{aligned}
\end{equation}


\section{Simulation} \label{simulation}
This section provides an extensive MC simulation study to evaluate the performance of the mentioned estimators given in Section \ref{Estimation methods}. In MC simulation study, we carry out the following settings:
\begin{itemize}
    \item The sample sizes; $n=50,100,200,500,1000$
\item The number of repetition: 5000
\item Initial values of parameters $\lambda$, $p_1$ and $p_2$  are taken as follows: \newline
$Senerio_{I} = \left( {\lambda  = 1.5, p_1=0.5, p_2 =0.3} \right)$,\\
$Senerio_{II} = \left( {\lambda  = 0.5, p_1=0.4, p_2 =0.5} \right)$,\\
$Senerio_{III} = \left( {\lambda  = 0.6, p_1=0.3, p_2 =0.2} \right)$,\\
$Senerio_{IV} = \left( {\lambda  = 2, p_1=0.2, p_2 =0.4} \right)$,\\ 
\end{itemize}

We utilize bias, mean squared error (MSE), and mean relative error (MRE) in order to evaluate the performance of the examined estimators. These measures are given by 

$$Bias=\frac{1}{5000}\sum\limits_{i=1}^{5000}\left( \hat{\Upsilon}-\Upsilon
\right),$$

$$MSE=\frac{1}{5000}\sum\limits_{i=1}^{5000}\left( \hat{\Upsilon%
}-\Upsilon \right) ^{2},$$

$$MRE=\frac{1}{5000}\sum_{i=1}^{5000}\frac{|\hat{\Upsilon}-\Upsilon|}{\Upsilon},$$
where $\Upsilon=\left(\lambda,p_1,p_2\right)$. 
\subsection{Random sample generation}
To generate random samples from the $CLRBTE(\lambda, p_1, p_2)$ distribution, we use an acceptance-rejection (AR) sampling algorithm.
We choose the Weibull distribution as the proposal distribution in the AR algorithm. The AR algorithm is given as follows:

\textbf{Algorithm 1.}

\textbf{A1.} Generate data on random variable $Y \sim Weibull(\gamma,\nu)$ 
with the PDF $g$ given as follows: 
\begin{equation*}
g\left( \gamma,\nu \right)=\gamma \nu x^{\nu -1}e^{-\gamma
x^\nu}.
\end{equation*}

\textbf{A2.} Generate $U$ from standard uniform distribution(independent of $%
Y$).

\textbf{A3.} If%
\begin{equation*}
U<\frac{f\left( Y;\lambda,p_1,p_2\right) }{k\times g\left(
Y;\gamma,\nu \right) }
\end{equation*}%
then set $X=Y$ (\textquotedblleft accept\textquotedblright ); otherwise go
back to A1 (\textquotedblleft reject\textquotedblright ), where the PDF $f$ $%
\left( .\right) $ is given as in Eq. (\ref{clrbte: pdf})  and 
\begin{equation*}
k=\underset{z\in 
\mathbb{R}
_{+}}{\max }\frac{f\left( z;\lambda,p_1,p_2\right) }{g\left(
z;\gamma,\nu \right) }.
\end{equation*}%

In MC simulation study, We use Algorithm 1 to generate random sample. Furthermore, all optimization problems are solved via \textsf{R} software \citep{Rsoftware} and \texttt{optim()} function, with the BFGS (Broyden--Fletcher--Goldfarb--Shanno) \cite{nocedal2006numerical} algorithm.

We report the findings of MC simulation study in Tables \ref{sim:tab1}-\ref{sim:tab4}.

\begin{table}[H]
\centering
\caption{The simulation results in $\lambda=1.5$, $p_1=0.5$ and $p_2=0.3$}
\scalebox{0.9} {\ \label{sim:tab1}
\begin{tabular}{ccccccccccc}\hline
          &      &         & Bias    &         &        & MSE    &        &        & MRE    &        \\\hline
Estimator & $n$    & $\hat\lambda$   & $\hat p_1$   & $\hat p_2$  & $\hat\lambda$  & $\hat p_1$   & $\hat p_2$ & $\hat\lambda$   & $\hat p_1$   & $\hat p_2$  \\\hline

MLE    & 50   & -0.0817 & -0.0702 & 0.0669 & 0.1155 & 0.0267 & 0.0327 & 0.1793 & 0.2541 & 0.4837 \\
       & 100  & -0.0773 & -0.0619 & 0.0567 & 0.0946 & 0.0251 & 0.0318 & 0.1675 & 0.2503 & 0.4736 \\
       & 200  & -0.0542 & -0.0387 & 0.0454 & 0.0791 & 0.0199 & 0.0290 & 0.1468 & 0.2168 & 0.4673 \\
       & 500  & -0.0165 & -0.0368 & 0.0377 & 0.0521 & 0.0174 & 0.0268 & 0.1158 & 0.2115 & 0.4287 \\
       & 1000 & 0.0054  & -0.0109 & 0.0068 & 0.0240 & 0.0090 & 0.0140 & 0.0782 & 0.1429 & 0.3122 \\\hline
LSE    & 50   & -0.3657 & -0.1937 & 0.0991 & 0.2744 & 0.0624 & 0.0485 & 0.3053 & 0.4519 & 0.6268 \\
       & 100  & -0.2759 & -0.1618 & 0.0811 & 0.2008 & 0.0557 & 0.0434 & 0.2593 & 0.4226 & 0.5871 \\
       & 200  & -0.1748 & -0.1070 & 0.0780 & 0.1299 & 0.0401 & 0.0427 & 0.2040 & 0.3398 & 0.5790 \\
       & 500  & -0.0970 & -0.0579 & 0.0437 & 0.0807 & 0.0273 & 0.0334 & 0.1579 & 0.2717 & 0.5083 \\
       & 1000 & -0.0357 & -0.0200 & 0.0082 & 0.0475 & 0.0160 & 0.0219 & 0.1159 & 0.2015 & 0.3982 \\\hline
WLSE   & 50   & -0.3604 & -0.1953 & 0.1177 & 0.2353 & 0.0609 & 0.0448 & 0.2806 & 0.4346 & 0.5916 \\
       & 100  & -0.2718 & -0.1625 & 0.1107 & 0.1761 & 0.0523 & 0.0402 & 0.2406 & 0.3957 & 0.5657 \\
       & 200  & -0.1767 & -0.1088 & 0.0867 & 0.1174 & 0.0382 & 0.0399 & 0.1867 & 0.3145 & 0.5457 \\
       & 500  & -0.0755 & -0.0475 & 0.0422 & 0.0570 & 0.0205 & 0.0271 & 0.1280 & 0.2272 & 0.4434 \\
       & 1000 & -0.0236 & -0.0135 & 0.0067 & 0.0292 & 0.0106 & 0.0162 & 0.0887 & 0.1594 & 0.3350 \\\hline
ADE    & 50   & -0.2347 & -0.1546 & 0.1243 & 0.1724 & 0.0467 & 0.0447 & 0.2366 & 0.3716 & 0.5845 \\
       & 100  & -0.2265 & -0.1477 & 0.1231 & 0.1491 & 0.0459 & 0.0416 & 0.2203 & 0.3703 & 0.5675 \\
       & 200  & -0.1563 & -0.1020 & 0.0901 & 0.1029 & 0.0351 & 0.0388 & 0.1773 & 0.3025 & 0.5340 \\
       & 500  & -0.0753 & -0.0487 & 0.0455 & 0.0562 & 0.0204 & 0.0270 & 0.1273 & 0.2267 & 0.4432 \\
       & 1000 & -0.0243 & -0.0146 & 0.0087 & 0.0294 & 0.0107 & 0.0161 & 0.0885 & 0.1591 & 0.3335 \\\hline
CvME   & 50   & -0.2392 & -0.1707 & 0.1427 & 0.1982 & 0.0523 & 0.0556 & 0.2525 & 0.4169 & 0.6639 \\
       & 100  & -0.2114 & -0.1487 & 0.1259 & 0.1660 & 0.0504 & 0.0538 & 0.2338 & 0.4032 & 0.6578 \\
       & 200  & -0.1445 & -0.1003 & 0.0883 & 0.1169 & 0.0377 & 0.0449 & 0.1936 & 0.3303 & 0.5870 \\
       & 500  & -0.0853 & -0.0552 & 0.0469 & 0.0769 & 0.0263 & 0.0334 & 0.1545 & 0.2677 & 0.5078 \\
       & 1000 & -0.0303 & -0.0188 & 0.0098 & 0.0462 & 0.0157 & 0.0219 & 0.1149 & 0.2000 & 0.3972 \\\hline
MPSE   & 50   & -0.3487 & -0.1682 & 0.0817 & 0.2479 & 0.0584 & 0.0382 & 0.2803 & 0.3977 & 0.5275 \\
       & 100  & -0.3225 & -0.1594 & 0.0741 & 0.2184 & 0.0534 & 0.0348 & 0.2621 & 0.3797 & 0.5038 \\
       & 200  & -0.2534 & -0.1390 & 0.0679 & 0.1704 & 0.0534 & 0.0343 & 0.2173 & 0.3608 & 0.4852 \\
       & 500  & -0.1553 & -0.0893 & 0.0270 & 0.1017 & 0.0371 & 0.0267 & 0.1559 & 0.2807 & 0.4469 \\
       & 1000 & -0.0704 & -0.0385 & 0.0198 & 0.0435 & 0.0159 & 0.0178 & 0.0961 & 0.1740 & 0.3449 \\\hline
TADE   & 50   & -0.2227 & -0.1321 & 0.0761 & 0.1725 & 0.0457 & 0.0410 & 0.2334 & 0.3567 & 0.5501 \\
       & 100  & -0.2175 & -0.1238 & 0.0661 & 0.1577 & 0.0422 & 0.0403 & 0.2234 & 0.3424 & 0.5480 \\
       & 200  & -0.1572 & -0.0953 & 0.0491 & 0.1087 & 0.0366 & 0.0367 & 0.1794 & 0.3069 & 0.5224 \\
       & 500  & -0.0765 & -0.0461 & 0.0349 & 0.0585 & 0.0221 & 0.0308 & 0.1268 & 0.2337 & 0.4756 \\
       & 1000 & -0.0292 & -0.0170 & 0.0104 & 0.0294 & 0.0113 & 0.0183 & 0.0876 & 0.1618 & 0.3518 \\\hline
MSADE  & 50   & -0.2380 & -0.1109 & 0.0482 & 0.1518 & 0.0349 & 0.0320 & 0.2077 & 0.2979 & 0.4817 \\
       & 100  & -0.2223 & -0.1037 & 0.0361 & 0.1447 & 0.0307 & 0.0289 & 0.2028 & 0.2843 & 0.4564 \\
       & 200  & -0.1595 & -0.0852 & 0.0283 & 0.0983 & 0.0304 & 0.0287 & 0.1610 & 0.2689 & 0.4559 \\
       & 500  & -0.0829 & -0.0500 & 0.0161 & 0.0599 & 0.0213 & 0.0276 & 0.1235 & 0.2211 & 0.4428 \\
       & 1000 & -0.0433 & -0.0244 & 0.0047 & 0.0351 & 0.0132 & 0.0202 & 0.0937 & 0.1706 & 0.3725 \\\hline
MSALDE & 50   & -0.2909 & -0.1551 & 0.0713 & 0.2344 & 0.0577 & 0.0430 & 0.2674 & 0.3980 & 0.5649 \\
       & 100  & -0.2840 & -0.1488 & 0.0674 & 0.2122 & 0.0526 & 0.0415 & 0.2556 & 0.3836 & 0.5440 \\
       & 200  & -0.2184 & -0.1224 & 0.0581 & 0.1580 & 0.0513 & 0.0388 & 0.2120 & 0.3589 & 0.5224 \\
       & 500  & -0.1339 & -0.0783 & 0.0392 & 0.0980 & 0.0353 & 0.0359 & 0.1578 & 0.2812 & 0.5185 \\
       & 1000 & -0.0650 & -0.0373 & 0.0254 & 0.0474 & 0.0177 & 0.0221 & 0.1039 & 0.1907 & 0.3868 \\\hline
\end{tabular}
}
\end{table}


\begin{table}[H]
\centering
\caption{The simulation results in $\lambda=0.5$, $p_1=0.4$ and $p_2=0.5$}
\scalebox{0.9} {\ \label{sim:tab2}
\begin{tabular}{ccccccccccc}\hline
          &      &         & Bias    &         &        & MSE    &        &        & MRE    &        \\\hline
Estimator & $n$    & $\hat\lambda$   & $\hat p_1$   & $\hat p_2$  & $\hat\lambda$  & $\hat p_1$   & $\hat p_2$ & $\hat\lambda$   & $\hat p_1$   & $\hat p_2$  \\\hline

MLE    & 50   & 0.0767  & 0.0598  & -0.0517 & 0.0519 & 0.0459 & 0.0778 & 0.3493 & 0.4384 & 0.4778 \\
       & 100  & 0.0456  & 0.0371  & -0.0505 & 0.0307 & 0.0413 & 0.0733 & 0.2731 & 0.4272 & 0.4621 \\
       & 200  & 0.0235  & 0.0317  & -0.0340 & 0.0205 & 0.0389 & 0.0672 & 0.2241 & 0.3982 & 0.4372 \\
       & 500  & 0.0149  & 0.0205  & -0.0332 & 0.0139 & 0.0355 & 0.0545 & 0.1908 & 0.3942 & 0.3887 \\
       & 1000 & 0.0024  & 0.0061  & -0.0282 & 0.0081 & 0.0245 & 0.0354 & 0.1426 & 0.3155 & 0.3072 \\\hline
LSE    & 50   & 0.0218  & 0.0605  & -0.3040 & 0.0731 & 0.1186 & 0.4328 & 0.4138 & 0.6430 & 1.0035 \\
       & 100  & 0.0213  & 0.0410  & -0.1742 & 0.0522 & 0.0848 & 0.2185 & 0.3389 & 0.5571 & 0.7228 \\
       & 200  & 0.0177  & 0.0383  & -0.1158 & 0.0316 & 0.0702 & 0.1491 & 0.2742 & 0.5394 & 0.6079 \\
       & 500  & 0.0136  & 0.0322  & -0.0845 & 0.0202 & 0.0484 & 0.0871 & 0.2200 & 0.4608 & 0.4771 \\
       & 1000 & 0.0030  & 0.0257  & -0.0601 & 0.0109 & 0.0314 & 0.0518 & 0.1738 & 0.3801 & 0.3739 \\\hline
WLSE   & 50   & -0.0593 & -0.0619 & -0.2289 & 0.0554 & 0.0910 & 0.4196 & 0.3751 & 0.5881 & 0.9259 \\
       & 100  & -0.0438 & -0.0357 & -0.0368 & 0.0319 & 0.0621 & 0.1445 & 0.2847 & 0.5139 & 0.6127 \\
       & 200  & -0.0231 & -0.0228 & -0.0131 & 0.0230 & 0.0589 & 0.1014 & 0.2438 & 0.5134 & 0.5244 \\
       & 500  & -0.0126 & -0.0221 & -0.0220 & 0.0155 & 0.0436 & 0.0662 & 0.2017 & 0.4424 & 0.4313 \\
       & 1000 & 0.0034  & 0.0093  & -0.0352 & 0.0082 & 0.0250 & 0.0398 & 0.1501 & 0.3350 & 0.3312 \\\hline
ADE    & 50   & -0.0144 & -0.0108 & -0.1276 & 0.0460 & 0.0837 & 0.2516 & 0.3388 & 0.5766 & 0.7054 \\
       & 100  & 0.0086  & 0.0103  & -0.0510 & 0.0332 & 0.0539 & 0.1112 & 0.2831 & 0.4779 & 0.5378 \\
       & 200  & 0.0017  & -0.0068 & -0.0451 & 0.0222 & 0.0526 & 0.0904 & 0.2362 & 0.4909 & 0.4975 \\
       & 500  & 0.0004  & 0.0057  & -0.0465 & 0.0144 & 0.0382 & 0.0633 & 0.1936 & 0.4183 & 0.4200 \\
       & 1000 & 0.0003  & 0.0017  & -0.0306 & 0.0080 & 0.0252 & 0.0392 & 0.1491 & 0.3357 & 0.3292 \\\hline
CvME   & 50   & 0.0909  & 0.0939  & -0.1706 & 0.0810 & 0.1318 & 0.3346 & 0.4138 & 0.6573 & 0.8770 \\
       & 100  & 0.0609  & 0.0588  & -0.1165 & 0.0538 & 0.0876 & 0.1916 & 0.3334 & 0.5486 & 0.6765 \\
       & 200  & 0.0458  & 0.0600  & -0.1042 & 0.0330 & 0.0703 & 0.1422 & 0.2730 & 0.5275 & 0.5949 \\
       & 500  & 0.0340  & 0.0418  & -0.0792 & 0.0203 & 0.0475 & 0.0853 & 0.2169 & 0.4512 & 0.4725 \\
       & 1000 & 0.0191  & 0.0308  & -0.0597 & 0.0108 & 0.0308 & 0.0509 & 0.1715 & 0.3742 & 0.3704 \\\hline
MPSE   & 50   & -0.1279 & -0.0713 & -0.6228 & 0.0693 & 0.1855 & 1.9005 & 0.4440 & 0.6591 & 1.4353 \\
       & 100  & -0.0941 & -0.0790 & -0.1894 & 0.0451 & 0.0530 & 0.2454 & 0.3498 & 0.5196 & 0.6977 \\
       & 200  & -0.0752 & -0.0608 & -0.0619 & 0.0318 & 0.0546 & 0.0978 & 0.2935 & 0.4965 & 0.4990 \\
       & 500  & -0.0511 & -0.0549 & -0.0328 & 0.0207 & 0.0473 & 0.0638 & 0.2358 & 0.4727 & 0.4138 \\
       & 1000 & -0.0397 & 0.0448  & 0.0109  & 0.0123 & 0.0358 & 0.0386 & 0.1783 & 0.3896 & 0.3290 \\\hline
TADE   & 50   & -0.0253 & -0.0232 & -0.2083 & 0.0572 & 0.0997 & 0.3458 & 0.3764 & 0.6191 & 0.8542 \\
       & 100  & -0.0257 & -0.0142 & -0.1014 & 0.0349 & 0.0674 & 0.1475 & 0.2963 & 0.5384 & 0.6186 \\
       & 200  & -0.0136 & -0.0085 & -0.0495 & 0.0232 & 0.0535 & 0.1045 & 0.2443 & 0.4992 & 0.5281 \\
       & 500  & -0.0083 & -0.0028 & -0.0473 & 0.0149 & 0.0418 & 0.0710 & 0.1998 & 0.4373 & 0.4344 \\
       & 1000 & -0.0053 & 0.0006  & -0.0298 & 0.0085 & 0.0264 & 0.0423 & 0.1513 & 0.3414 & 0.3373 \\\hline
MSADE  & 50   & -0.0638 & -0.0235 & -0.2072 & 0.0427 & 0.0427 & 0.1803 & 0.3405 & 0.3807 & 0.5908 \\
       & 100  & -0.0389 & -0.0203 & -0.1011 & 0.0281 & 0.0313 & 0.0959 & 0.2736 & 0.3782 & 0.4394 \\
       & 200  & -0.0288 & -0.0110 & -0.0814 & 0.0217 & 0.0392 & 0.0755 & 0.2367 & 0.3395 & 0.3937 \\
       & 500  & -0.0173 & -0.0103 & -0.0527 & 0.0146 & 0.0305 & 0.0508 & 0.1922 & 0.3375 & 0.3337 \\
       & 1000 & -0.0137 & -0.0084 & -0.0296 & 0.0087 & 0.0236 & 0.0354 & 0.1466 & 0.2974 & 0.2850 \\\hline
MSALDE & 50   & -0.0876 & -0.0546 & -0.3317 & 0.0638 & 0.0740 & 0.3293 & 0.4247 & 0.5390 & 0.8593 \\
       & 100  & -0.0740 & -0.0535 & -0.1621 & 0.0422 & 0.0596 & 0.1872 & 0.3373 & 0.5314 & 0.6540 \\
       & 200  & -0.0589 & -0.0522 & -0.0812 & 0.0318 & 0.0611 & 0.1131 & 0.2943 & 0.5086 & 0.5318 \\
       & 500  & -0.0458 & -0.0375 & -0.0408 & 0.0217 & 0.0524 & 0.0783 & 0.2427 & 0.4922 & 0.4525 \\
       & 1000 & -0.0306 & -0.0144 & -0.0108 & 0.0126 & 0.0368 & 0.0453 & 0.1796 & 0.3964 & 0.3506 \\

          \hline
\end{tabular}
}
\end{table}

\begin{table}[H]
\centering
\caption{The simulation results in $\lambda=0.6$, $p_1=0.3$ and $p_2=0.2$}
\scalebox{0.9} {\ \label{sim:tab3}
\begin{tabular}{ccccccccccc}\hline
          &      &         & Bias    &         &        & MSE    &        &        & MRE    &        \\\hline
Estimator & $n$    & $\hat\lambda$   & $\hat p_1$   & $\hat p_2$  & $\hat\lambda$  & $\hat p_1$   & $\hat p_2$ & $\hat\lambda$   & $\hat p_1$   & $\hat p_2$  \\\hline

MLE       
          & 50  & 0.1098  & 0.0340  & 0.1071  & 0.0647 & 0.0243 & 0.0620 & 0.3208 & 0.4098 & 0.9359 \\
          & 100 & 0.0404  & 0.0038  & 0.0484  & 0.0336 & 0.0189 & 0.0336 & 0.2377 & 0.3827 & 0.7089 \\
          & 200 & 0.0110  & -0.0011 & 0.0050  & 0.0230 & 0.0148 & 0.0196 & 0.1972 & 0.3287 & 0.5569 \\
          & 500 & -0.0041 & -0.0047 & -0.0141 & 0.0127 & 0.0105 & 0.0118 & 0.1456 & 0.2716 & 0.4443 \\
         & 1000 & -0.002 & -0.003 & -0.0031 & 0.0055 & 0.0048 & 0.0059 & 0.0969 & 0.1821 & 0.3132 \\
          \hline
LSE       
          & 50  & 0.0450  & 0.0323  & -0.0771 & 0.1165 & 0.0581 & 0.2139 & 0.4481 & 0.5733 & 1.9004 \\
          & 100 & 0.0176  & 0.0093  & -0.0602 & 0.0754 & 0.0418 & 0.1257 & 0.3591 & 0.5382 & 1.4345 \\
          & 200 & 0.0181  & 0.0217  & -0.0820 & 0.0563 & 0.0329 & 0.0616 & 0.3121 & 0.4902 & 0.9838 \\
          & 500 & -0.0053 & 0.0045  & -0.0669 & 0.0329 & 0.0202 & 0.0265 & 0.2477 & 0.4000 & 0.6508 \\
         & 1000 & -0.0056 & -0.0002 & -0.0317 & 0.0170 & 0.0115 & 0.0085 & 0.1782 & 0.2995 & 0.3765 \\
          \hline
WLSE      
          & 50  & -0.0417 & -0.0475 & -0.0075 & 0.0778 & 0.0334 & 0.1620 & 0.3782 & 0.4911 & 1.6246 \\
          & 100 & -0.0368 & -0.0364 & -0.0245 & 0.0547 & 0.0315 & 0.0814 & 0.3099 & 0.4875 & 1.1433 \\
          & 200 & -0.0115 & -0.0088 & -0.0330 & 0.0368 & 0.0234 & 0.0350 & 0.2504 & 0.4218 & 0.7422 \\
          & 500 & -0.0108 & -0.0066 & -0.0264 & 0.0177 & 0.0138 & 0.0158 & 0.1789 & 0.3226 & 0.5017 \\
          & 1000 & -0.0055 & -0.0037 & -0.0111 & 0.0086 & 0.0072 & 0.0069 & 0.1247 & 0.2316 & 0.3354 \\
          \hline
ADE       
          & 50  & 0.0199  & -0.0215 & 0.0617  & 0.0668 & 0.0336 & 0.1216 & 0.3382 & 0.5014 & 1.3324 \\
          & 100 & -0.0255 & -0.0381 & 0.0105  & 0.0446 & 0.0277 & 0.0690 & 0.2805 & 0.4684 & 1.0297 \\
          & 200 & -0.0107 & -0.0111 & -0.0220 & 0.0315 & 0.0208 & 0.0322 & 0.2327 & 0.3994 & 0.7100 \\
          & 500 & -0.0127 & -0.0097 & -0.0219 & 0.0168 & 0.0133 & 0.0154 & 0.1743 & 0.3163 & 0.4963 \\
          & 1000 & -0.0073 & -0.0058 & -0.0092 & 0.0084 & 0.0071 & 0.0068 & 0.1232 & 0.2291 & 0.3322 \\
          \hline
CvME      
          & 50  & 0.1134  & 0.0235  & 0.0759  & 0.1145 & 0.0584 & 0.1729 & 0.4288 & 0.5783 & 1.6924 \\
          & 100 & 0.0484  & 0.0056  & 0.0037  & 0.0752 & 0.0419 & 0.1150 & 0.3527 & 0.5370 & 1.3690 \\
          & 200 & 0.0341  & 0.0196  & -0.0457 & 0.0545 & 0.0327 & 0.0539 & 0.3031 & 0.4873 & 0.9240 \\
          & 500 & 0.0001  & 0.0040  & -0.0558 & 0.0324 & 0.0200 & 0.0255 & 0.2448 & 0.3977 & 0.6344 \\
          & 1000 & -0.0031 & -0.0007 & -0.0263 & 0.0167 & 0.0114 & 0.0083 & 0.1768 & 0.2984 & 0.3709 \\
          \hline
MPSE      
          & 50  & -0.1164 & -0.0455 & -0.1610 & 0.0717 & 0.0262 & 0.1832 & 0.3758 & 0.4530 & 1.4599 \\
          & 100 & -0.1280 & -0.0761 & -0.1132 & 0.0568 & 0.0242 & 0.0721 & 0.3315 & 0.4483 & 1.0265 \\
          & 200 & -0.0950 & -0.0621 & -0.0648 & 0.0390 & 0.0210 & 0.0301 & 0.2645 & 0.3989 & 0.6808 \\
          & 500 & -0.0566 & -0.0381 & -0.0366 & 0.0204 & 0.0141 & 0.0130 & 0.1822 & 0.3125 & 0.4600 \\
          & 1000 & -0.0307 & -0.0239 & -0.0044 & 0.0080 & 0.0066 & 0.0058 & 0.1141 & 0.2072 & 0.3068 \\
          \hline
TADE      
          & 50  & -0.0268 & -0.0413 & 0.0233  & 0.0561 & 0.0322 & 0.1066 & 0.3207 & 0.5061 & 1.2465 \\
          & 100 & -0.0371 & -0.0410 & -0.0052 & 0.0437 & 0.0255 & 0.0580 & 0.2795 & 0.4538 & 0.9305 \\
          & 200 & -0.0365 & -0.0288 & -0.0277 & 0.0307 & 0.0208 & 0.0289 & 0.2327 & 0.3978 & 0.6663 \\
          & 500 & -0.0222 & -0.0155 & -0.0264 & 0.0168 & 0.0128 & 0.0138 & 0.1710 & 0.3060 & 0.4767 \\
          & 1000 & -0.0092 & -0.0073 & -0.0076 & 0.0071 & 0.0059 & 0.0060 & 0.1096 & 0.2016 & 0.3150 \\
          \hline
MSADE     
          & 50  & -0.0404 & -0.0025 & -0.0760 & 0.0492 & 0.0323 & 0.1075 & 0.2982 & 0.4616 & 1.3072 \\
          & 100 & -0.0581 & -0.0343 & -0.0401 & 0.0363 & 0.0238 & 0.0646 & 0.2583 & 0.4217 & 1.0053 \\
          & 200 & -0.0472 & -0.0317 & -0.0329 & 0.0294 & 0.0209 & 0.0392 & 0.2309 & 0.3936 & 0.7724 \\
          & 500 & -0.0280 & -0.0176 & -0.0304 & 0.0178 & 0.0151 & 0.0210 & 0.1768 & 0.3252 & 0.5727 \\
         & 1000 & -0.0228 & -0.0180 & -0.0058 & 0.0092 & 0.0077 & 0.0085 & 0.1242 & 0.2281 & 0.3698 \\

          \hline
MSALDE    
          & 50  & -0.0728 & -0.0224 & -0.1349 & 0.0697 & 0.0365 & 0.1467 & 0.3672 & 0.5174 & 1.5429 \\
          & 100 & -0.0970 & -0.0635 & -0.0812 & 0.0527 & 0.0297 & 0.0821 & 0.3152 & 0.4940 & 1.1565 \\
          & 200 & -0.0830 & -0.0532 & -0.0657 & 0.0405 & 0.0267 & 0.0460 & 0.2741 & 0.4610 & 0.8558 \\
          & 500 & -0.0452 & -0.0293 & -0.0405 & 0.0217 & 0.0164 & 0.0197 & 0.1897 & 0.3427 & 0.5615 \\
          & 1000 & -0.0351 & -0.0273 & -0.0059 & 0.0106 & 0.0086 & 0.0076 & 0.1310 & 0.2368 & 0.3474 \\
          \hline
\end{tabular}
}
\end{table}
\begin{table}[H]
\centering
\caption{The simulation results in $\lambda=2$, $p_1=0.2$ and $p_2=0.4$}
\scalebox{0.9} {\ \label{sim:tab4}
\begin{tabular}{ccccccccccc}\hline
          &      &         & Bias    &         &        & MSE    &        &        & MRE    &        \\\hline
Estimator & $n$    & $\hat\lambda$   & $\hat p_1$   & $\hat p_2$  & $\hat\lambda$  & $\hat p_1$   & $\hat p_2$ & $\hat\lambda$   & $\hat p_1$   & $\hat p_2$  \\\hline
MLE       
          & 50  & 0.3739  & 0.0849  & -0.0092 & 0.5916 & 0.0279 & 0.0682 & 0.3066 & 0.6328 & 0.5579 \\
          & 100 & 0.2538  & 0.0714  & -0.0598 & 0.4427 & 0.0252 & 0.0500 & 0.2684 & 0.6194 & 0.4746 \\
          & 200 & 0.1603  & 0.0513  & -0.0706 & 0.3391 & 0.0208 & 0.0411 & 0.2320 & 0.5703 & 0.4262 \\
          & 500 & 0.0809  & 0.0334  & -0.0654 & 0.2728 & 0.0161 & 0.0232 & 0.2089 & 0.5053 & 0.3044 \\
          & 1000 & 0.0679  & 0.0252  & -0.0429 & 0.1440 & 0.0087 & 0.0100 & 0.1505 & 0.3669 & 0.1982 \\
\hline
LSE       
          & 50  & 0.2185  & 0.0803  & -0.1574 & 1.2412 & 0.0623 & 0.2830 & 0.4263 & 0.8045 & 1.0584 \\
          & 100 & 0.2043  & 0.0645  & -0.1153 & 0.8197 & 0.0427 & 0.1469 & 0.3518 & 0.6961 & 0.7636 \\
          & 200 & 0.1600  & 0.0601  & -0.1365 & 0.6762 & 0.0374 & 0.1150 & 0.3227 & 0.6945 & 0.6657 \\
          & 500 & 0.1341  & 0.0528  & -0.1197 & 0.5249 & 0.0280 & 0.0757 & 0.2809 & 0.6138 & 0.5034 \\
           & 1000 & 0.1445  & 0.0403  & -0.0562 & 0.2723 & 0.0165 & 0.0196 & 0.2026 & 0.4824 & 0.2789 \\
\hline
WLSE      
          & 50  & -0.0138 & 0.0300  & -0.1280 & 0.8810 & 0.0328 & 0.2221 & 0.3707 & 0.5816 & 0.9652 \\
          & 100 & 0.0819  & 0.0433  & -0.1097 & 0.6464 & 0.0287 & 0.1197 & 0.3169 & 0.5999 & 0.7078 \\
          & 200 & 0.1006  & 0.0437  & -0.1027 & 0.4807 & 0.0264 & 0.0748 & 0.2750 & 0.6152 & 0.5624 \\
          & 500 & 0.0703  & 0.0317  & -0.0769 & 0.3521 & 0.0208 & 0.0369 & 0.2307 & 0.5480 & 0.3764 \\
           & 1000 & 0.0635  & 0.0237  & -0.0463 & 0.1817 & 0.0113 & 0.0144 & 0.1679 & 0.4178 & 0.2374 \\
\hline
ADE      
          & 50  & 0.1077  & 0.0339  & -0.0496 & 0.6554 & 0.0285 & 0.1483 & 0.3175 & 0.5847 & 0.7773 \\
          & 100 & 0.0713  & 0.0341  & -0.0851 & 0.5225 & 0.0246 & 0.1048 & 0.2895 & 0.5810 & 0.6553 \\
          & 200 & 0.0968  & 0.0399  & -0.0906 & 0.4217 & 0.0239 & 0.0688 & 0.2586 & 0.5962 & 0.5375 \\
          & 500 & 0.0513  & 0.0277  & -0.0775 & 0.3405 & 0.0198 & 0.0379 & 0.2275 & 0.5422 & 0.3764 \\
           & 1000 & 0.0609  & 0.0225  & -0.0443 & 0.1739 & 0.0108 & 0.0139 & 0.1641 & 0.4090 & 0.2321 \\
\hline
CvME     
          & 50  & 0.3957  & 0.0633  & -0.0264 & 1.3051 & 0.0589 & 0.2246 & 0.4308 & 0.8068 & 0.9437 \\
          & 100 & 0.3181  & 0.0632  & -0.0553 & 0.8453 & 0.0419 & 0.1288 & 0.3550 & 0.7006 & 0.7086 \\
          & 200 & 0.2017  & 0.0544  & -0.0999 & 0.6622 & 0.0366 & 0.0971 & 0.3196 & 0.6887 & 0.6241 \\
          & 500 & 0.1525  & 0.0526  & -0.1113 & 0.5289 & 0.0276 & 0.0755 & 0.2821 & 0.6089 & 0.4974 \\
           & 1000 & 0.1658  & 0.0424  & -0.0500 & 0.2705 & 0.0164 & 0.0185 & 0.2018 & 0.4788 & 0.2714 \\
\hline
MPSE     
          & 50  & -0.5210 & 0.0142  & -0.4051 & 0.8229 & 0.0165 & 0.4023 & 0.3785 & 0.4820 & 1.2314 \\
          & 100 & -0.3548 & 0.0002  & -0.2645 & 0.6410 & 0.0176 & 0.1715 & 0.3365 & 0.5406 & 0.8481 \\
          & 200 & -0.2902 & -0.0137 & -0.1899 & 0.5168 & 0.0177 & 0.0965 & 0.2952 & 0.5660 & 0.6305 \\
          & 500 & -0.2358 & -0.0209 & -0.1266 & 0.4067 & 0.0165 & 0.0463 & 0.2605 & 0.5411 & 0.4237 \\
           & 1000 & -0.1580 & -0.0177 & -0.0719 & 0.2383 & 0.0108 & 0.0173 & 0.1945 & 0.4311 & 0.2520 \\
\hline
TADE      
          & 50  & -0.0654 & 0.0141  & -0.0991 & 0.5646 & 0.0218 & 0.1713 & 0.3021 & 0.5547 & 0.8243 \\
          & 100 & -0.0328 & 0.0229  & -0.1151 & 0.4629 & 0.0228 & 0.1014 & 0.2770 & 0.5810 & 0.6521 \\
          & 200 & -0.0241 & 0.0177  & -0.1016 & 0.3794 & 0.0209 & 0.0673 & 0.2444 & 0.5799 & 0.5269 \\
          & 500 & -0.0286 & 0.0132  & -0.0847 & 0.3148 & 0.0174 & 0.0365 & 0.2252 & 0.5375 & 0.3689 \\
           & 1000 & $3 \times 10^{-5}$  & 0.0132  & -0.0581 & 0.1892 & 0.0106 & 0.0144 & 0.1730 & 0.4183 & 0.2358 \\
\hline
MSADE    
          & 50  & -0.1474 & 0.0307  & -0.1295 & 0.1748 & 0.0186 & 0.1480 & 0.1364 & 0.4697 & 0.7401 \\
          & 100 & -0.0891 & 0.0192  & -0.0855 & 0.1405 & 0.0167 & 0.0938 & 0.1187 & 0.4569 & 0.6069 \\
          & 200 & -0.0839 & 0.0071  & -0.0732 & 0.1464 & 0.0135 & 0.0583 & 0.1220 & 0.4203 & 0.4879 \\
          & 500 & -0.0588 & 0.0028  & -0.0533 & 0.1317 & 0.0099 & 0.0287 & 0.1234 & 0.3744 & 0.3417 \\
           & 1000 & -0.0378 & 0.0025  & -0.0368 & 0.0919 & 0.0067 & 0.0135 & 0.1036 & 0.3079 & 0.2315 \\
\hline
MSALDE    
          & 50  & -0.3717 & 0.0242  & -0.3079 & 0.6407 & 0.0240 & 0.2884 & 0.3251 & 0.5221 & 1.0745 \\
          & 100 & -0.2395 & 0.0164  & -0.2321 & 0.5815 & 0.0264 & 0.1710 & 0.3151 & 0.6082 & 0.8341 \\
          & 200 & -0.2253 & -0.0055 & -0.1653 & 0.4640 & 0.0202 & 0.0973 & 0.2749 & 0.5792 & 0.6325 \\
          & 500 & -0.1577 & -0.0068 & -0.1130 & 0.3512 & 0.0172 & 0.0455 & 0.2402 & 0.5417 & 0.4201 \\
 & 1000 & -0.1008 & -0.0073 & -0.0602 & 0.1893 & 0.0102 & 0.0163 & 0.1736 & 0.4106 & 0.2515 \\
          \hline
\end{tabular}
}
\end{table}

\begin{itemize}

\item
The findings of Table \ref{sim:tab1} ($\lambda=1.5$, $p_1=0.5$, $p_2=0.3$) show that all estimators improve systematically as the sample size increases, with decreases observed in bias, MSE, and MRE. The MLE provides satisfactory performance overall and becomes highly efficient for large samples; however, several robust estimators provide competitive or superior accuracy, particularly in smaller sample sizes. Among these, the MSADE estimator shines by maintaining consistently low bias and stable MSE values across all parameters. The ADE, WLSE, and TADE estimators also show balanced and reliable performance, outperforming LSE in small-sample settings. On the other hand, the LSE and MSALDE produce higher bias and larger MRE, reflecting weaker finite-sample behavior, while MPSE performs moderately but remains inferior to MSADE and ADE. Overall, although the MLE becomes highly efficient as $n$ increases, the robust MSADE presents more stable and accurate estimates for small and medium sample sizes.

\item
Table~\ref{sim:tab2} provides the simulation results under the setting $\lambda = 0.5$, $p_1 = 0.4$, and $p_2 = 0.5$, demonstrate that all estimators improve steadily with increasing sample size, as reflected by consistent decreases in bias, MSE, and MRE. Among the competing techniques, the MSADE presents the most favorable performance, achieving relatively small bias and the lowest MSE and MRE values across all parameters, particularly for medium and large samples. The ADE, WLSE, and TADE also exhibit stable and reliable estimates, often outperforming the MLE in small sample sizes. Although the MLE becomes increasingly competitive as the sample size increases, it provides higher bias and MSE for small samples. Conversely, the LSE, MPSE, and MSALDE show weaker finite-sample behavior, characterized by larger bias and substantially higher MSE and MRE values, especially for the estimation of $p_2$. 

\item
In Table~\ref{sim:tab3},the simulation results based on $\lambda = 0.6$, $p_1 = 0.3$, and $p_2 = 0.2$ are given. The results indicate that all estimators exhibit improved performance as the sample size increases, with notable decreases in bias, MSE, and MRE. Among the competing methods, the MSADE exhibits consistently favorable results, yielding comparatively low bias and stable MSE values across all parameters, particularly for medium and large sample sizes. The ADE, WLSE, and TADE also provide reliable behavior, often surpassing the MLE in small sample sizes. Although the MLE becomes increasingly accurate as $n$ increases, it demonstrates relatively higher bias and MSE levels for small samples, especially in estimating $p_2$. On the other hand, the LSE, MPSE, and MSALDE tend to suffer from larger bias and elevated MSE and MRE values, indicating weaker finite-sample performance. Overall, the findings emphasize the robustness and efficiency of the MSADE estimator under this parameter setting.

\item
Table~\ref{sim:tab4} demonstrates the simulation results under the setting $\lambda  = 2, p_1=0.2, p_2 =0.4$, show that all estimators improve as the sample size increases, as evidenced by decreasing bias, MSE, and MRE values. Among the examined methods, the MSADE provides the best overall performance, achieving consistently lower bias and MSE across all parameters and sample sizes, which indicates strong robustness and efficiency. The ADE and WLSE also perform reliably, frequently outperforming the classical MLE, particularly in smaller samples. Although the MLE becomes competitive as the sample size increases, it presents relatively higher bias and MSE in small sample sizes. On the other hand, the MPSE and MSALDE demonstrate poor finite-sample behavior, characterized by large biases and comparatively high MSE and MRE values. In general, the findings highlight MSADE as the most effective estimator under the simulation scenario.

\item
The simulation results across all considered settings consistently indicate that the performance of all estimators improves as the sample size increases, with decreases in bias, MSE, and MRE. Among the estimators, the MSADE method emerges as the most robust and efficient, providing consistently low bias and stable MSE values across all parameters, particularly in medium and large sample sizes. The ADE, WLSE, and TADE also provide reliable and balanced performance, often surpassing the MLE in small sample sizes. The MLE itself becomes highly efficient for large samples, but it tends to show higher bias and MRE in smaller samples, especially for parameters like $p_2$. In contrast, the LSE, MPSE, and MSALDE generally show weaker finite-sample behavior, with larger bias and higher MSE and MRE values. In general, for practitioners seeking robust estimation in small to medium samples, we recommend the MSADE for all parameters, while ADE, WLSE, and TADE provide viable alternatives. For very large samples, MLE achieves competitive efficiency and can be reliably used.

\end{itemize}

\section{Real data analysis}
\label{sec:real}
To demonstrate the utility of the CLRBTE distribution, we present three real-world data examples in this section. We also compare the fits to those of some of its competitors, including the exponential, transmuted exponential (TE) \cite{aryal2011transmuted}, and transmuted generalized Rayleigh (TGR) \cite{merovci2014transmuted} distributions.  As a result, we take into account the following seven distinct selection criteria:  Cram\'{e}r-von-Mises statistics (CvM), Anderson-Darling statistics (AD), Kolmogorov-Smirnov statistics (KS), Akaike's information criterion (AIC), and associated p-values.  The PDFs of fitted models to data sets are shown in Table \ref{tab:disttable}.
\begin{table}[H]
\caption{The list of fitted PDFs to real-world datasets}
\label{tab:disttable}
\scalebox{0.9}{
\begin{tabular}{lll}
\hline
\textbf{Distribution} & \textbf{PDF} & \textbf{Range of Parameters} \\
\hline
E & $f_{E}(x;\lambda)
=\displaystyle \lambda e^{-\lambda x}$ 
& $\lambda>0$ \\
& & \\[0.4em]
TE & $f_{TE}(x;\lambda,\theta)
=\displaystyle (1 + \theta)\lambda e^{-\lambda x} 
- 2\theta (1 - e^{-\lambda x}) \lambda e^{-\lambda x}$ 
& $\lambda>0,\ -1<\theta<1$ \\
& & \\[0.4em]
TGR & $f_{TGR}(x;\lambda, \beta,\theta)
=\displaystyle 2 \lambda \beta^2 x e^{-(\beta x)^2}
\left( 1 - e^{-(\beta x)^2} \right)^{\lambda - 1}
\left[ 1 + \theta - 2\theta 
\left( 1 - e^{-(\beta x)^2} \right)^{\lambda} \right]$ 
& $\lambda,\beta>0,\ -1<\theta<1$ \\
\hline
\end{tabular}
}
\end{table}

\subsection{Survival time data}
In this subsection, we analyze a well-known right-skewed sample of survival times (in days) for patients diagnosed with acute myelogenous leukemia, originally reported in \cite{feigl1965estimation}. The survival time dataset is as follows:  
65, 156, 100, 134, 16, 108, 121, 4, 39, 143, 56, 26, 22, 1, 1, 5, 65, 56, 65, 17, 7, 16, 22, 3, 4,  
2, 3, 8, 4, 3, 30, 4, 43.

Table \ref{tab:real1est} shows the MLEs and corresponding standard errors (SEs) of the parameters for the fitted models while Table \ref{tab:real1comp} provides the selection criteria for the survival time data. The fitted CDF and PDF plots are shown in Figures \ref{Fig: real1cdf} and \ref{Fig: real1pdf} respectively.

\begin{figure} [H]
    \centering
    \includegraphics[width=0.8\linewidth]{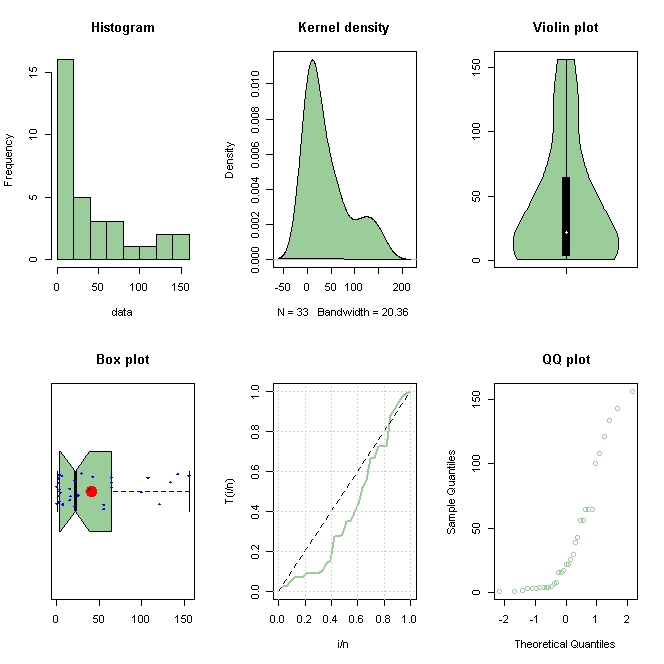}
    \caption{Non-parametric plots for survival time data}
    \label{Fig: nonplot1}
\end{figure}

Figure \ref{Fig: nonplot1} illustrates a variety of non-parametric graphical methods utilized for survival time data.  The data distribution exhibits significant skewness, with most observations concentrated at lower levels.  The kernel density and violin plots demonstrate multimodality, signifying various substructures within the data.  The Q-Q plot significantly deviates from the reference line, indicating a deficiency in normality, however the box plot corroborates the presence of outliers.

\begin{table}[H]
\centering
\caption{The comparison statistics for the survival time data}
\label{tab:real1comp}
\scalebox{0.8}{
\begin{tabular}{cccccccc}
\hline
Distribution & AIC & KS & AD & CvM & p-value(KS) & p-value(AD) & p-value(CvM) \\
\hline
CLRBTE & 312.4142 & 0.1044 & 0.4992 & 0.0643 & 0.8646 & 0.7464 & 0.7901 \\
TE     & 313.7042 & 0.2028 & 1.7879 & 0.2542 & 0.1325 & 0.1208 & 0.1832 \\
E      & 312.9003 & 0.2182 & 2.3066 & 0.3250 & 0.0863 & 0.0631 & 0.1148 \\
TGR    & 312.8599 & 0.1318 & 0.6735 & 0.1003 & 0.6156 & 0.5801 & 0.5862 \\
\hline
\end{tabular}
}
\end{table}


\begin{table}[H]
\centering
\caption{The MLEs and SEs of the parameters for the survival time data}
\label{tab:real1est}
\scalebox{0.75}{
\begin{tabular}{cccccccccccc}
\hline
Distribution & $\hat{\lambda}$ & $\hat{p_1}$ & $\hat{p_2}$ & $\hat{\theta}$ & $\hat{\beta}$ & SE($\hat{\lambda}$) & SE($\hat{p_1}$) & SE($\hat{p_2}$) & SE($\hat{\theta}$) & SE($\hat{\beta}$) \\
\hline
CLRBTE & 0.0174 & 0.6369 & 0.0878 & - & - & 0.0055 & 0.3458 & 0.6688 & - & - \\
TE     & 0.0206 & - & - & 0.3704 & - & 0.0057 & - & - & 0.3617 & - \\
E      & 0.0245 & - & - & - & - & 0.0043 & - & - & - & - \\
TGR    & 0.3066 & - & - & 0.4216 & 0.0089 & 0.0608 & - & - & 0.4719 & 0.0025 \\
\hline
\end{tabular}
}
\end{table}

\begin{figure}[H]
    \centering
        \includegraphics[width=0.6\linewidth]{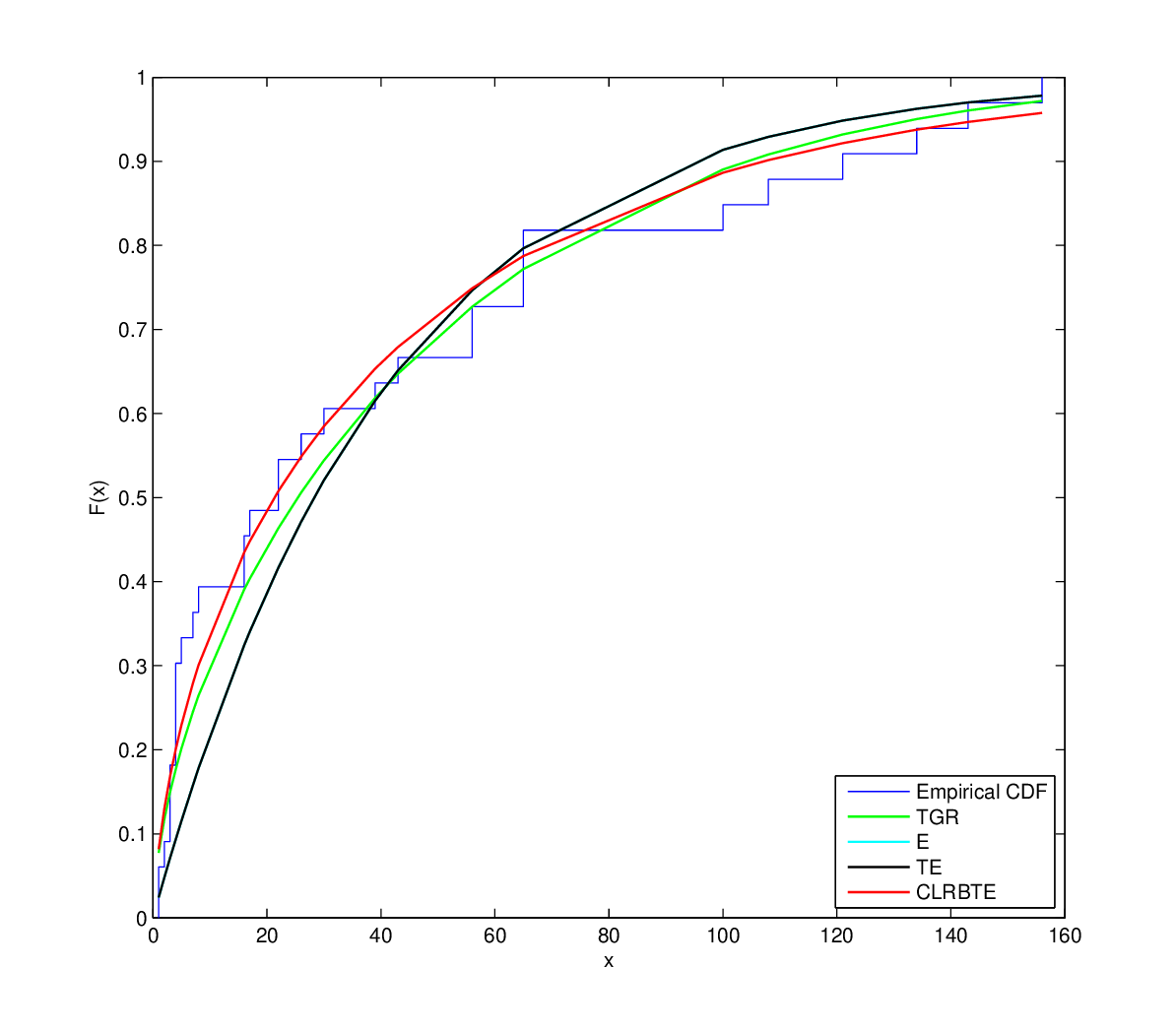}
    \caption{The fitted CDFs for the survival time data}
    \label{Fig: real1cdf}
\end{figure}


\begin{figure}[H]
    \centering
        \includegraphics[width=0.6\linewidth]{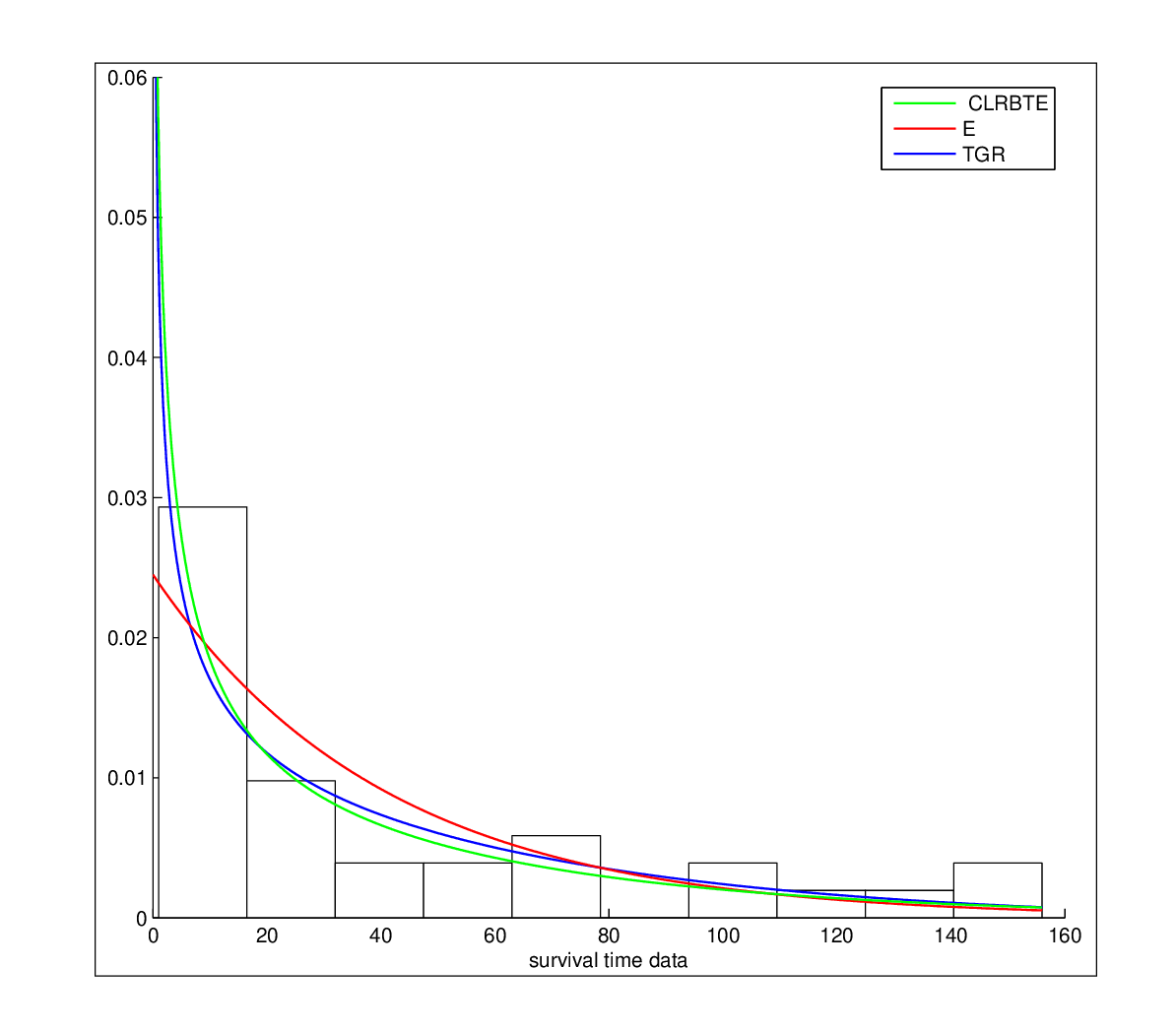}
    \caption{The fitted PDFs for the survival time data}
    \label{Fig: real1pdf}
\end{figure}

Based on all of the selection criteria listed in Table \ref{tab:real1comp}, we identify the CLRBTE distribution as the model that fits the survival time dataset the best out of all the models.  It is evident from Figures \ref{Fig: real1cdf} and \ref{Fig: real1pdf} that the closest model to the survival time dataset is the CLRBTE distribution.

\subsection{Failure time data}
In this subsection, we analyze the data set provided by \cite{murthy2004weibull}, denotes the failure times of 20 components. The failure time data are:
0.0003, 0.0298, 0.1648, 0.3529, 0.4044, 0.5712, 0.5808, 0.7607, 0.8188, 1.1296, 1.2228, 1.2773, 1.9115, 2.2333, 2.3791, 3.0916, 3.4999, 3.7744, 7.4339, 13.6866.

Table \ref{tab:real2comp} provides the selection criteria while Table \ref{tab:real2est} gives the MLEs and their SEs for the failure time data. Furthermore, the fitted CDF and PDF plots are shown in Figures \ref{Fig: real2cdf} and \ref{Fig: real2pdf} respectively.

\begin{figure} [H]
    \centering
    \includegraphics[width=0.8\linewidth]{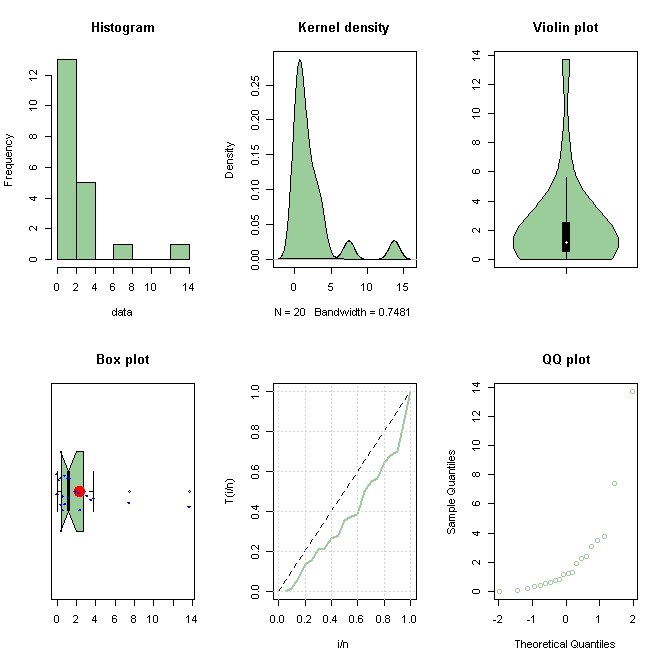}
    \caption{Non-parametric plots for failure time data}
    \label{Fig: nonplot2}
\end{figure}

Failure time data nonparametric summaries are shown in Figure~\ref{Fig: nonplot2}. From Figure \ref{Fig: nonplot2}, the failure dataset displays a heavy right skew. The density estimate indicates a dominant peak alongside secondary fluctuations, consistent with the characteristics of the potential mixture.


\begin{table}[H]
\centering
\caption{The comparison statistics for the failure time data}
\label{tab:real2comp}
\scalebox{0.8}{
\begin{tabular}{cccccccc}
\hline
Distribution & AIC & KS & AD & CvM & p-value(KS) & p-value(AD) & p-value(CvM) \\
\hline
CLRBTE & 74.2783 & 0.1126 & 0.3306 & 0.0345 & 0.9371 & 0.9126 & 0.9627 \\
TE     & 75.1337 & 0.1241 & 0.6187 & 0.0463 & 0.8810 & 0.6283 & 0.9032 \\
E      & 74.7239 & 0.1691 & 0.9941 & 0.1104 & 0.5597 & 0.3593 & 0.5401 \\
TGR    & 73.6468 & 0.1274 & 0.3710 & 0.0545 & 0.8619 & 0.8756 & 0.8540 \\
\hline
\end{tabular}
}
\end{table}


\begin{table}[H]
\centering
\caption{The MLEs and SEs of the parameters for the failure time data}
\label{tab:real2est}
\scalebox{0.75}{
\begin{tabular}{cccccccccccc}
\hline
Distribution & $\hat{\lambda}$ & $\hat{p_1}$ & $\hat{p_2}$ & $\hat{\theta}$ & $\hat{\beta}$ & SE($\hat{\lambda}$) & SE($\hat{p_1}$) & SE($\hat{p_2}$) & SE($\hat{\theta}$) & SE($\hat{\beta}$) \\
\hline
CLRBTE & 0.2098 & 0.2510 & 0.5475 & - & - & 0.2231 & 0.6935 & 0.6300 & - & - \\
TE     & 0.3213 & - & - & 0.6236 & - & 0.1115 & - & - & 0.4071 & - \\
E      & 0.4413 & - & - & - & - & 0.0987 & - & - & - & - \\
TGR    & 0.2717 & - & - & 0.7501 & 0.1099 & 0.0578 & - & - & 0.4221 & 0.0483 \\
\hline
\end{tabular}
}
\end{table}

\begin{figure}[H]
    \centering
        \includegraphics[width=0.6\linewidth]{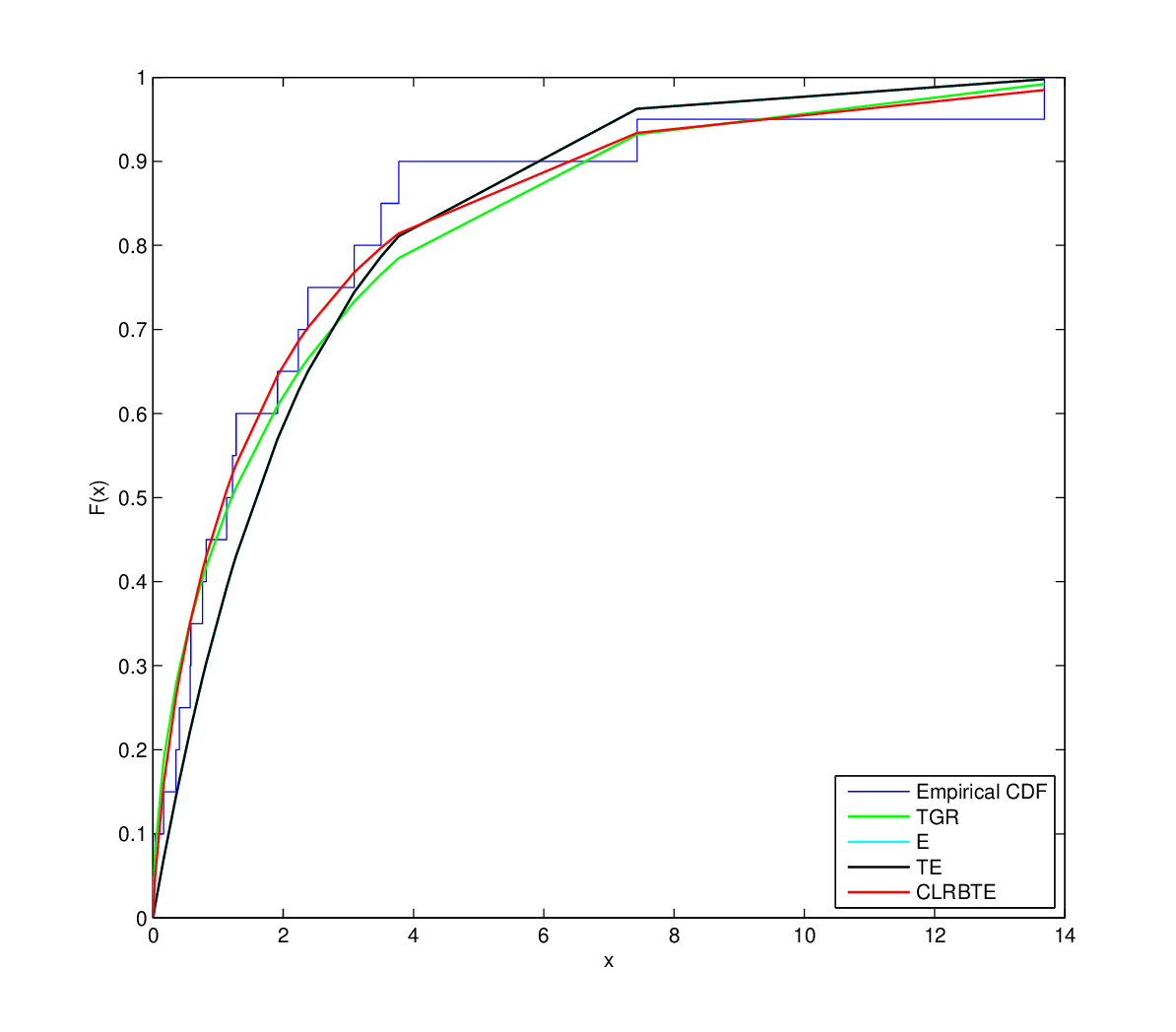}
    \caption{The fitted CDFs for the failure time data}
    \label{Fig: real2cdf}
\end{figure}


\begin{figure}[H]
    \centering
        \includegraphics[width=0.6\linewidth]{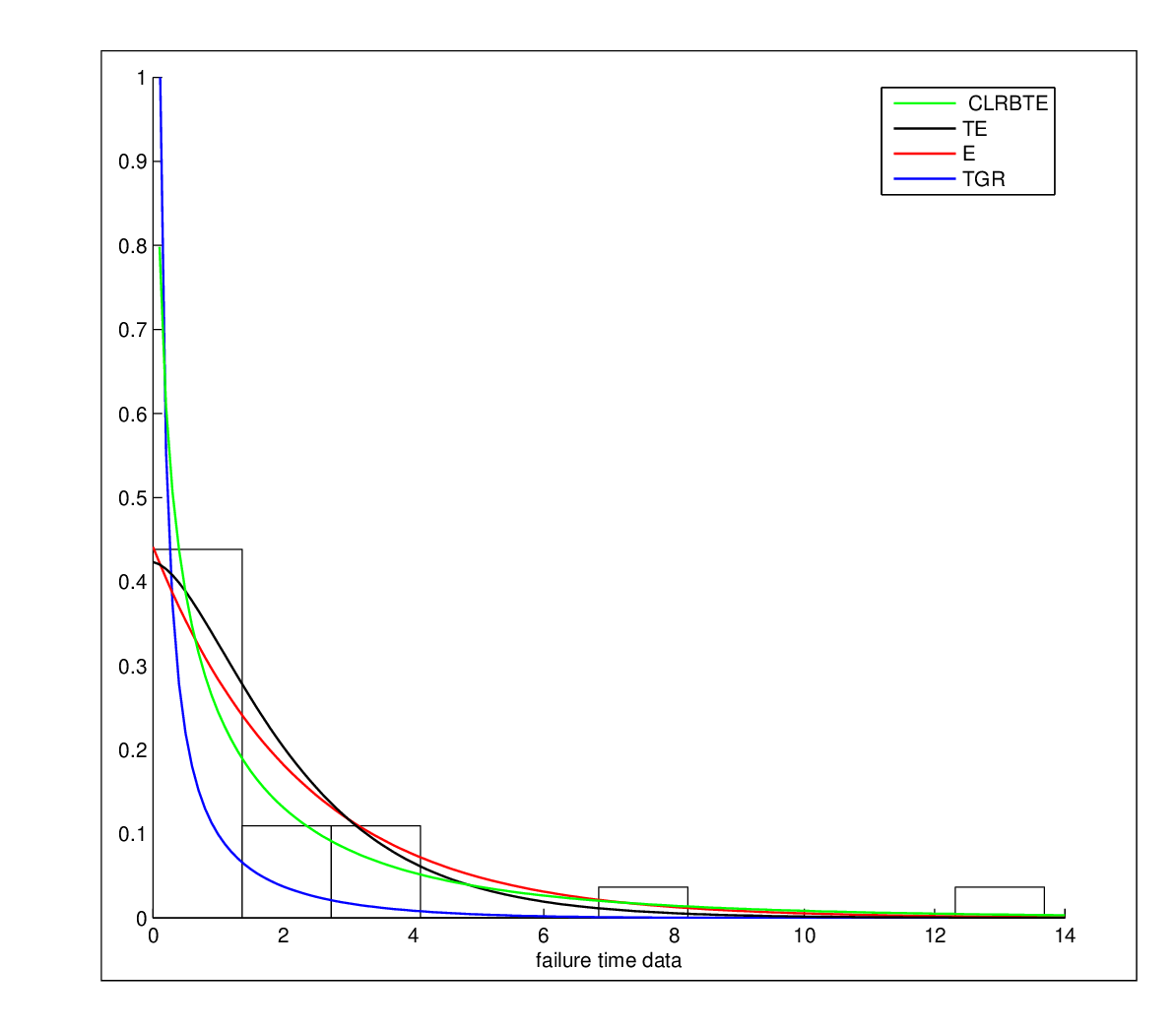}
    \caption{The fitted PDFs for the failure time data}
    \label{Fig: real2pdf}
\end{figure}

From Figures \ref{Fig: real2cdf} and \ref{Fig: real2pdf}, we observe that the best-fitted model is CLRBTE distribution for the failure dataset.

\section{Conclusion}
\label{conc}
Through present study, we introduce a new family of distributions based the distributions of the first three lower record values to the literature. We also examine in detail a special case of the proposed  family based on the exponential distribution. We present some statistical properties of the CLRBTE distribution, including point estimation, a simulation study, and its application to real-life data. Nine different estimators have been proposed for point estimation. A MC simulation study has been designed to compare the performance of these estimators based on bias, MSE, and MRE across four different parameter settings, varying sample sizes, and a 5000-iteration. As a result of the simulation study, it has been concluded that MSADE is a good alternative to MLE in estimating the parameters of the CLRBTE distribution. Subsequently, the CLRBTE distribution fit both datasets better than potential competitors E, TE, and TGR, according to seven different comparison criteria in the two real data analyses. In future studies, new members can be proposed for the suggested family of distributions.

\bigskip
\noindent\textbf{Funding} This paper is not financially supported.
\newline
\bigskip
\noindent\textbf{Data Availability} The reference list contains the publicly available datasets.

\bigskip
\noindent\textbf{Conflict of interest} Not applicable.

\bibliographystyle{apalike} 
\bibliography{ref}
\end{document}